\documentclass[letterpaper, 9 pt, journal]{IEEEtran}  
\pdfoutput=1
\usepackage[pdftex]{graphicx}                                      
\usepackage{graphicx}
\usepackage{amsmath, amssymb, commath}
\usepackage{array}
\usepackage{subfigure}
\usepackage{paralist}
\usepackage{cite}
\usepackage{acronym}
\usepackage[usenames,dvipsnames]{color}
\usepackage{algorithm}
\usepackage{algorithmic}
\usepackage{bm}
 \newtheorem{proposition}{Proposition}

 \newtheorem{remark}{Remark}
 
 \newtheorem{problem}{Problem}

\IEEEoverridecommandlockouts                              
\usepackage{graphics} 

\let\oldmarginpar\marginpar
\renewcommand\marginpar[1]{\-\oldmarginpar[\raggedleft\footnotesize #1]%
{\raggedright\footnotesize #1}}

\title{\LARGE \bf
Stochastic receding horizon control of nonlinear stochastic systems with probabilistic state constraints
}
\author{Shridhar K. Shah, Herbert G. Tanner and Chetan D. Pahlajani
\thanks{This work is supported by ARL MAST CTA \# W911NF-08-2-0004}
\thanks{Shridhar K. Shah and Herbert G. Tanner are with Department of Mechanical Engineering,
        University of Delaware, Newark, DE, USA
        {\tt\small \{shridhar, btanner\}@udel.edu}}%
\thanks{Chetan D. Pahlajani is with Department of Mathematical Sciences,
        University of Delaware, Newark, DE, USA
        {\tt\small chetan@math.udel.edu}}%
\thanks{A portion of this work has been previously presented at International Conference on Robotics and Automation (ICRA) 2012 \cite{shah-ICRA2012}, which dealt with systems without control multiplicative term and unbounded inputs and a linear example. We extend the theory to systems with control multiplicative term. We also comment on bounded input case and explain a recovery strategy. A nonlinear example, numerical solution methods, extra results and simulations are included.}
}

\begin{document}

\acrodef{sde}[\textsc{sde}]{stochastic differential equation}
\acrodef{pde}[\textsc{pde}]{partial differential equation}
\acrodef{hjb}[\textsc{hjb}]{Hamilton-Jacobi-Bellman}
\acrodef{mpc}[\textsc{mpc}]{model predictive control}
\acrodef{rhc}[\textsc{rhc}]{receding horizon control}
\acrodef{gshs}[\textsc{gshs}]{gereral stochastic hybrid system}

\maketitle
\thispagestyle{empty}
\pagestyle{empty}
\begin{abstract}

The paper describes a receding horizon control design framework for continuous-time stochastic nonlinear systems subject to probabilistic state constraints.  The intention is to derive solutions that are implementable in real-time on currently available mobile processors.  The approach consists of decomposing the problem into designing receding horizon reference paths based on the drift component of the system dynamics, and then implementing a stochastic optimal controller to allow the system to stay close and follow the reference path.  In some cases, the stochastic optimal controller can be obtained in closed form; in more general cases, pre-computed numerical solutions can be implemented in real-time without the need for on-line computation. The convergence of the closed loop system is established assuming no constraints on control inputs, and simulation results are provided to corroborate the theoretical predictions.

Keywords - stochastic model predictive control, nonlinear systems, exit time, stochastic optimal control, path integral

\end{abstract}

\section{Introduction}

The behavior of robotic systems can be uncertain due to a variety of reasons, including 
noise in sensor measurements and environmental effects.  
Such effects are often represented by stochastic models (for example, ocean waves \cite{ochi2005ocean}, wind 
guests \cite{barr1974wind} and uneven terrain\cite{Kewlani2010}). 
For nonlinear stochastic systems, existing methods for constrained optimal control are too computationally
demanding for real-time implementation.  Specifically, no real-time solution exists for continuous-time
nonlinear  stochastic systems with probabilistic state constraints. 
A receding horizon formulation partially lifts some of the computational burden associated with 
the nonlinear stochastic optimal control problem, but current state of the art does not allow real-time
implementation on processors at the low-end of the frequency scale.
This paper proposes a solution through a stochastic receding horizon formulation that is real-time 
implementable for nonlinear systems of modest dimension, and comes with probabilistic guarantees 
of convergence and state constraint satisfaction.

Within a predictive control framework, uncertainty can be accounted for by
either approximating sets that bound the system's trajectories 
\cite{Ram-02, Ram-06, Dar-09, Carson-08, Mar-02, Survey, MPC-Survey} 
or by stochastic models, with the latter having some specific advantages.  In particular,
while methods based on set-bounded models may result in 
over-conservative designs since they plan for the worst case,
the use of probabilistic constraints in the methods which are based on stochastic models, on the
other hand, allows for less conservatism.  In addition, stochastic model-based methods
provide some flexibility by
allowing one to adjust the probability that problem constraints are violated.
These two qualities enable stochastic model-based methods to offer solutions where
set-bounded methods may fail.

The structure of the dynamics, whenever it can be exploited, can greatly facilitate the solution of a 
\ac{mpc} problem.   When the stochastic dynamics is \emph{linear}, one may choose to
apply a Kalman filter or its variants and solve an iterative LQG problem \cite{Todorov2005}.
Alternatively, for linear stochastic systems, the optimal control 
problem under probabilistic constraints is tackled within a \emph{chance-constrained} model 
predictive control framework 
\cite{Chatterjee2011, Chatterjee2011-2, Hokayem2012, CCMPC99, CCMPC00, Lygeros09, Blackmore2010, Blackmore2011}. 
Chance-constraint formulations are available for linear discrete time systems with Gaussian noise 
\cite{Lygeros09, Cannon2007, Cannon2009, Cannon2011, Primbs2012,Primbs2010, Primbs2009, Schwarm1999, Hessem2003, hessem2004stochastic, Blackmore2011}.   

While methods exist to enable \ac{mpc} in linear stochastic systems \cite{Lygeros09, Cannon2007, Cannon2009, Cannon2011, Primbs2012,Primbs2010, Primbs2009, Schwarm1999, Hessem2003, hessem2004stochastic, Blackmore2011}, for most nonlinear systems, the stochastic receding horizon optimal control problem can not be solved in real-time. 
For example, a particle filter implementation of chance-constrained model predictive control is available for linear systems with probabilistic noise \cite{Blackmore06aprobabilistic, Blackmore2010}, and it is in principle applicable to nonlinear systems too.  However, the approximate solutions obtained using this method depend on the number of particles, and convergence is achieved after a sufficiently large number of particles is used.
Alternative (discrete-time) methods combine a hybrid density filter with dynamic programming \cite{Weissel},
the latter being the natural discrete formulation of the optimal control problem.  
In the hybrid systems
literature we find reach-avoid formulations of this problem \cite{HSCC2011, Automatica2010}, in 
which the indicator function of hitting goal or obstacle sets appears in the cost of the optimal control problem
(similarly to what is done in this paper).   Computational complexity currently limits the 
application of these methods to systems with up to three states \cite{Automatica2010}, while 
requirements for real-time implementation are not imposed.  
Invariably, computational complexity and accuracy issues surface in all discrete-time and space methods, 
either primarily due to the use of filters, or simply due to the resolution required in the
time or state-space domains. 


Time and space-discretization may be avoided if the problem is formulated in continuous space
and time.
Continuous-time solutions to stochastic optimal control problems are available 
for systems affine in control and with state independent and time invariant control transition matrix,
and it is based on path integrals \cite{Kappen2005-2}.   A path integral is essentially the solution
to a \ac{hjb} equation, obtained after the application of a particular transformation \cite{Fleming}.
In certain cases, the path integral is computable numerically using Laplace approximations or Monte Carlo 
sampling. Different applications of path-integral stochastic optimal control have been explored, such as 
reinforcement learning \cite{Theodorou-PI2}, variable stiffness control (equivalent to automatic tuning of PD 
gains) \cite{Buchli2011} and risk sensitive control \cite{Kappen2011}.  The main issue with path integrals 
is that for most nonlinear systems the solution is computationally demanding and can not be obtained 
in real-time on existing processors. This limits the application of path integral to real-time receding horizon control on miniature robots. 

The main contribution of this paper is to synthesize a real-time design for stochastic (receding horizon) control,
following an \emph{exit time} \cite{Day} formulation of the stochastic optimal control problem, instead of one 
based on path integrals.  The proposed formulation yields a time invariant control vector field, which is 
optimal in terms of actuation utilization.  What enables real-time implementation is the fact that
the field can be computed off-line and used on-line in a recursive manner.  
The formulation is based on a combination of deterministic planning with stochastic optimal control, where 
successive locally optimal stochastic controls are used to steer a system along a deterministic receding horizon 
reference trajectory, which is conceptually similar to Differential Dynamic Programming 
\cite{jacobson1970differential} and iterative LQG \cite{Todorov2005}.
While such a two-level planning and control strategies has been used successfully in 
a \emph{deterministic} setting \cite{burridge-ijrr-1999,Lamiraux,Pathak} there is no stochastic analog yet except our own work  \cite{shah-ICRA2012}.
Due to the explicit consideration of stochasticity, the proposed method offers almost sure 
(with probability one) guarantees of collision avoidance and convergence to a desired region, which
are elusive in a deterministic setting.

The work presented in this paper is organized in the following way. Section \ref{section:problemstatement} 
states the problem formally followed by an intuitive explanation of our approach in section 
\ref{section:intuitiveexample}. Section \ref{section:dayscontroller} explains a stochastic optimal control 
design, which is at the heart of our framework. Section \ref{section:framework} presents the design of the 
stochastic receding horizon framework and discusses the existence of solutions for our closed loop system. 
The convergence properties of the resulting stochastic hybrid system are established in 
Section \ref{section:stability}, and the issue of input saturation is brought up.
Section \ref{section:examples} offers
examples of linear and nonlinear systems, presents
simulation results for the cases of unbounded and bounded inputs, and discusses computation
methods for complex nonlinear stochastic systems. We conclude in Section \ref{section:conclusions}.


\section{Problem Statement}
\label{section:problemstatement}

Consider an uncertain dynamical system evolving within an open bounded 
region $\mathcal{S} \subset \mathbb{R}^n$.  Within $\mathcal{S}$, there is a 
closed set $\mathcal{O} \subset \mathcal{S}$ which represents forbidden 
areas (\emph{obstacles}). In that sense, the system can safely evolve only in the \emph{free
workspace} $\mathcal{P} \triangleq \mathcal{S}\setminus\mathcal{O}$. 

The dynamics of the system is given in the form of a 
\ac{sde} 
\begin{multline}
\label{eq:problemsystem}
\dif \mathrm{q}(t) = b\del[1]{\mathrm{q}(t)}\dif t + G\del[1]{\mathrm{q}(t)} \left[
u \del[1]{\mathrm{q}(t)} \dif t \right. \\
+\left. \Sigma \del[1]{\mathrm{q}(t)} \dif W(t)\right] \ , \ \mathrm{q}(0)  = \mathrm{q}_0
\end{multline}
where $\mathrm{q} \in \mathbb{R}^n$ is the \emph{state}, $b:\mathbb{R}^n \rightarrow \mathbb{R}^n$ 
is the drift term, $G: \mathbb{R}^n \rightarrow \mathbb{R}^m$ is the matrix of control 
vector fields, $u: \mathbb{R}^n \rightarrow \mathbb{R}^m$ is the \emph{control input}, 
and $\Sigma: \mathbb{R}^n \rightarrow \mathbb{R}^{m \times m}$ is the diffusion term. 
Let $W=\{W(t), \mathcal{F}_t:0 \le t < \infty\}$ be an $m$-dimensional Wiener process on 
the probability space $(\Omega,\mathcal{F},\mathbb{P})$, where $\Omega$ is the sample 
space, $\mathcal{F}$ is a $\sigma$-algebra on $\Omega$, $\mathbb{P}$ is the probability measure 
and $\{\mathcal{F}_t: t \geq 0\}$ is the filtration (i.e. an increasing family of sub-$\sigma$-algebras 
of $\mathcal{F}$) that is right continuous and $\mathcal{F}_0$ contains all $\mathbb{P}$-null 
sets.\footnote{The justification and the detailed definition for these mathematical constructions 
can be found in \cite{Karatzas}.} 

In a typical stochastic optimal control problem, one has to find a 
control sequence to steer the dynamics to a desired configuration, 
while minimizing the cost functional
\begin{align*}
&V(\mathrm{q},u) = \min_{u(t)} \mathbb{E} \sbr[3]{ \int_{0}^{\infty} L\del[1]{\mathrm{q}(s),u(s)}\dif s \; \Big| \; \mathrm{q}(0) = \mathrm{q}}\\
&\text{subject to } \mathbb{P}[\mathrm{q}(t) \in \mathcal{O}] = 0, \quad \forall t
\end{align*}
where the function $L$ is the incremental cost, assumed 
positive definite.


For general nonlinear systems, global analytic solutions to the above stochastic optimization problem 
are not available.  Numerical solutions can be obtained, but depending on the size of the dynamics
and the constraints of the problem, the computation cost can be too high for real-time implementation on processors on the lower side of the frequency scale.  This limitation motivates us to seek sub-optimal
solutions to the above problem by solving the following relaxation instead.
\begin{figure}[h]
      \centering
      \includegraphics[width=0.4\textwidth]{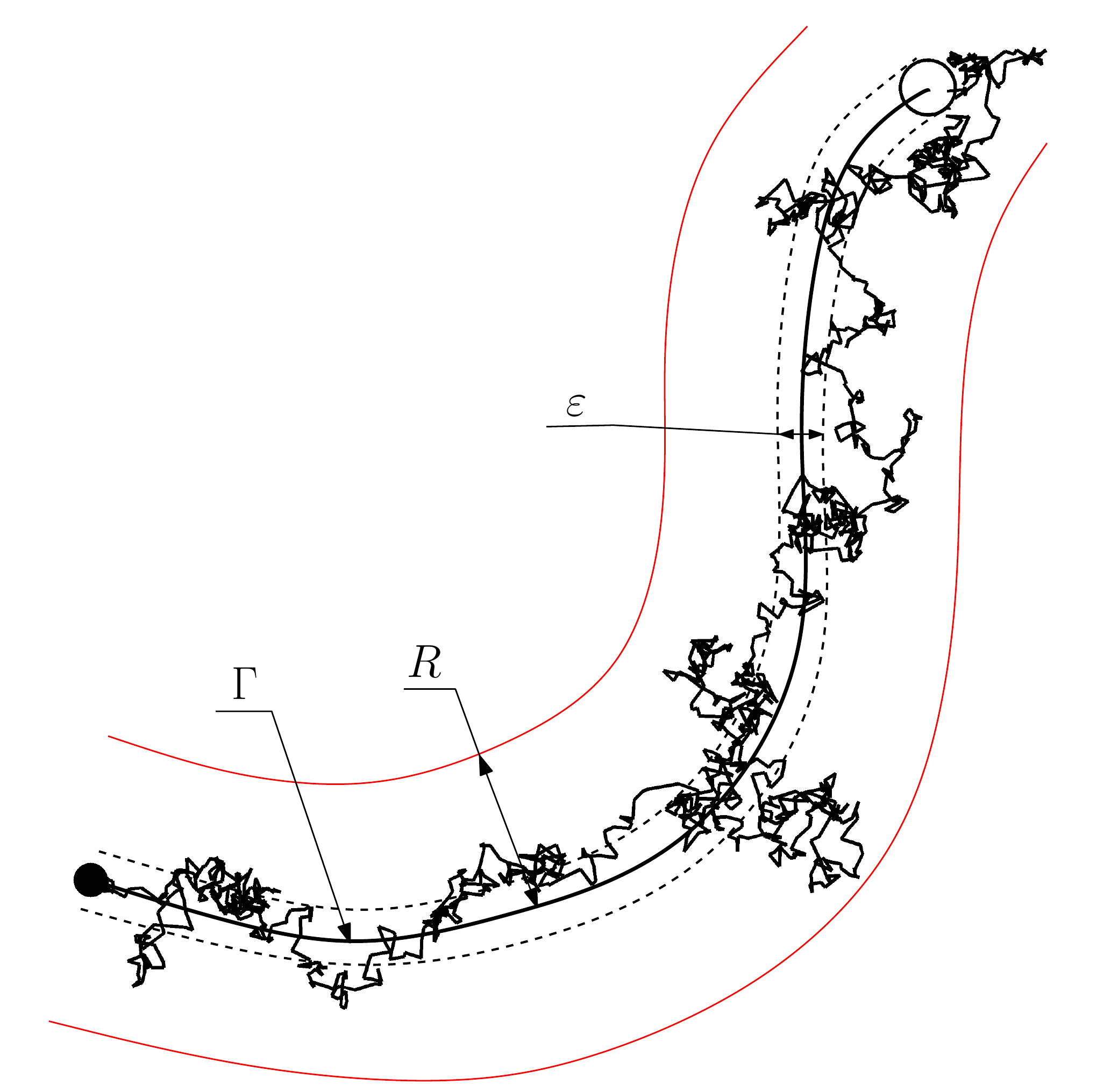}
      \caption{Illustration of the modified problem statement.  Obtain the solution $\Gamma$ to a 
      deterministic optimal control problem using the drift part of the dynamics
      (solid thick curve), and then maintain the full stochastic dynamics (thin sample path) $R$-close to that
      reference solution with accuracy $\varepsilon$.
\label{fig:problem_example}}
   \end{figure} 

\begin{problem}[Modified Problem Statement]
Find a sequence of feedback control laws $\{u_i(\mathrm{q})\}_{i=1}^N$ for \eqref{eq:problemsystem},
 such that  if $ \hat{\mathrm{q}}^*(t)$ is the solution of the system\footnote{Assume that the dimension of the controllability distribution is of rank $n$.}
\begin{equation} \label{deterministic dynamics}
\dot{\hat{\mathrm{q}}} = b(\hat{\mathrm{q}}) + G(\hat{\mathrm{q}}) u(\hat{\mathrm{q}})
\end{equation}
for a $\hat{u}^\ast(t)$ that minimizes the functional
\begin{align}
\label{eq:deterministic-cost}
&J(\mathrm{q},\hat{u}) = \min_{ \hat{u}} \int_{0}^{\infty} L\big(\hat{\mathrm{q}}(s), \hat{u}(s)\big) \dif s  \\
\nonumber & \text{subject to} \enspace 
\inf_{z \in \mathcal{O}, t>0} \| \hat{\mathrm{q}}(t) - z \| > R > 2\varepsilon > 0 \ , \qquad\hat{\mathrm{q}}(0) = \mathrm{q}  
\enspace. \nonumber
\end{align}
where, $R$ and $\varepsilon$ are positive constants. If $\Gamma = \{ \gamma \in \mathbb{R}^n \mid \exists t \in \mathbb{R} ; \gamma = \hat{\mathrm{q}}^*(t)\}$ denotes the
locus (path) of that solution, then for a given selection $\{\gamma_i\}_{i=1}^N \subset \Gamma$ of $N$ 
points on $\Gamma$ such that $\inf_{i,j} \| \gamma_i - \gamma_j\| > 2\varepsilon $, $\sup_{i,j} \| \gamma_i - \gamma_j\| < R - 2\varepsilon $ and $\hat{\mathrm{q}}_N = 0$, 
the application of $\{u_i(\mathrm{q})\}$ 
to \eqref{eq:problemsystem} results in sample paths $\mathrm{q}(t)$ that achieve
\begin{enumerate}[(i)]
\item $\mathbb{P} \left[ \inf_{\gamma \in \Gamma} \| \mathrm{q}(t) - \gamma \| < R\right] = 1,\; \forall t >0$ (almost-sure safety);
\item $\mathbb{P} \left[ \exists\, t_s < \infty : \|\mathrm{q}_N - \mathrm{q}(t_s) \| < \varepsilon \right] = 1$
(almost-sure convergence with accuracy $\varepsilon >0$);
\item $\mathbb{E} \left[ \int_{t_{i-1}}^{t_{i}} L\big(\mathrm{q}(s),u_i(s)\big)\dif s + \Phi\big(\mathrm{q}(t_{i})\big) \right]$
is minimized, 
where $t_{i-1}$ and $t_i$ are the first times $\mathrm{q}(t)$ enters an $\varepsilon$-neighborhood
of $\gamma_{i-1}$ and $\gamma_i$, respectively, and $\Phi(\mathrm{q}): \mathbb{R}^n \to \mathbb{R}_+$
is a \emph{terminal cost} function (local optimality).
\end{enumerate}
\end{problem}
\vspace{2ex}

Even in this form, the problem does not lend itself to efficiently computed solutions because
of the nonlinear infinite-horizon optimal control problem that needs to be solved to obtain $\Gamma$.
For this reason, the solution $(\hat{\mathrm{q}}^\ast,\hat{u}^\ast)$ of the deterministic optimal control
problem will be approximated by the solution of the receding horizon problem
\begin{subequations} \label{standard-rh}
\begin{align} \label{deterministic-rh}
&J_T(\mathrm{q},u_\mathsf{rh}) 
= \min_{u(t)} \int_{0}^{T} L\big(z(s), u(s)\big)\dif s + Q\big(z(T)\big) \\
 & \text{subject to} \enspace 
\dot{z} = b(z) + G(z)u\ , \qquad z(0) = \mathrm{q} \label{deterministic-dynamics}
\end{align}
\end{subequations}
where $T$ is the prediction horizon of the optimization, function $L$ is the same as in 
\eqref{eq:deterministic-cost}, and $Q: \mathbb{R}^n \to \mathbb{R}_+$ is
the terminal cost which approximates the truncated tail of the integral in
\eqref{eq:deterministic-cost}.  
The idea behind a receding horizon optimization strategy is that one solves
the finite horizon optimal control problem and obtains a control law 
$u_\mathsf{rh}(t)$ computed for $z(0) = \hat{\mathrm{q}}(t_0)$.
Control law $u_\mathsf{rh}(t)$ is applied on \eqref{deterministic dynamics} for the time interval 
$[t_0,t_1]$, $t_1 < t_0 + T$, during which time a new control law is 
computed for $z(0) = \hat{\mathrm{q}}(t_1)$, with $\hat{\mathrm{q}}(t_1)$ predicted based on
\eqref{deterministic dynamics}.  At time $t=t_1$,
the control law is updated and the process is repeated.
It is known \cite{jad-01} that if $Q(z)$ is 
a control Lyapunov function for \eqref{deterministic-dynamics}, and
\begin{equation} \label{condition-of-ali}
\min_{u}\big\{\dot{Q}(z) + L(z,u) \big\} \le - \eta(\|z\|) \enspace,
\end{equation}
where $\eta$ is a class-$\mathcal{K}$ function of $\|z\|$, then application
of $u_\mathsf{rh}(t)$ results in $\|z\| \to 0$ asymptotically with time.
We assume that $Q$ is a control Lyapunov function for \eqref{deterministic-dynamics} here as well,
and that there exists a positive definite function $\eta$ satisfying \eqref{condition-of-ali}.
In the our modified problem setting, $\{u_i(\mathrm{q})\}$ takes the 
place of $u_\mathsf{rh}(t)$ and $\hat{\mathrm{q}}(t_i) \equiv \{\gamma_i\}$.
%
%

\section{An intuitive example}
\label{section:intuitiveexample}

Consider a robot moving in a two-dimensional space, and described 
by single integrator dynamics perturbed by stochastic noise:
\begin{equation}
\label{eq:examplesystem}
\dif \mathrm{q} (t) = u\big(\mathrm{q}(t)\big) \dif t + \dif W(t); \enspace \mathrm{q}(0) = \mathrm{q}_0
\end{equation}
where $\mathrm{q} = [x \ y]^\intercal$ is the state vector, $u(\mathrm{q})$ is the control input and $W(t)$ is 
a two-dimensional Wiener process.  The objective is to find a feedback control law $u(\mathrm{q})$ to 
drive the system $\varepsilon$-close to the origin, while avoiding the boundary of a circle 
with radius $R$, centered at the origin. 

An obvious control strategy is to just steer the system along a direction toward the origin. 
A normalized vector pointing to the origin from the current state $\mathrm{q}$ is 
$- \frac{\mathrm{q}}{\|\mathrm{q} \|}$.  To satisfy the state constraints, the system should be 
forced away from the circle with radius $R$. 
One way to achieve this is by weighting the control input by a factor 
$\frac{1}{R - \|\mathrm{q}\|}$. This results in
\begin{equation}
\label{eq:intuitive-control}
u(\mathrm{q}) = - \frac{\mathrm{q}}{(R - \|\mathrm{q}\|)\|\mathrm{q} \| } \enspace.
\end{equation}
It turns out, this intuitive design yields a stochastic control law which is actually optimal. In fact, 
\eqref{eq:intuitive-control} minimizes the cost
\[V(\mathrm{q},u) = \mathbb{E} \left[  \int_0^{\tau} \frac{1}{2}u^\intercal\big(\mathrm{q}(s)\big) \,
u\big(\mathrm{q}(s)\big) \dif s + \Phi\big(\mathrm{q}(\tau)\big) \; \Big| \; \mathrm{q}(0) = \mathrm{q}\right]\]
where
\begin{align*}
\Phi\big(\mathrm{q}(\tau)\big)  &= 
\begin{cases} 0 & \text{on } {\| \mathrm{q}(\tau)\| = \varepsilon}\\
\infty &\text{on } {\| \mathrm{q}(\tau)\| = R} 
\end{cases} 
\end{align*}
and $\tau$ is the first time the state hits either the circle with radius $\varepsilon$ or that with radius $R$. 
Control law \eqref{eq:intuitive-control} guarantees that the system avoids the $R$-radius circle boundary
with probability one, and consequently hits the $\varepsilon$-radius circle with probability one,
because it is known that it almost surely exits the domain $\{\varepsilon < \|\mathrm{q}\|
< R \}$ somewhere (see \cite[Lemma 7.4]{Karatzas}, and the discussion in the section that follows).
Sample paths for the given controller are shown in Fig. \ref{fig:intuitive-example-single} for 
different initial conditions. 

\begin{figure}[ht]
\centering
\subfigure[Sample Paths]{
          \includegraphics[width=0.44\textwidth]{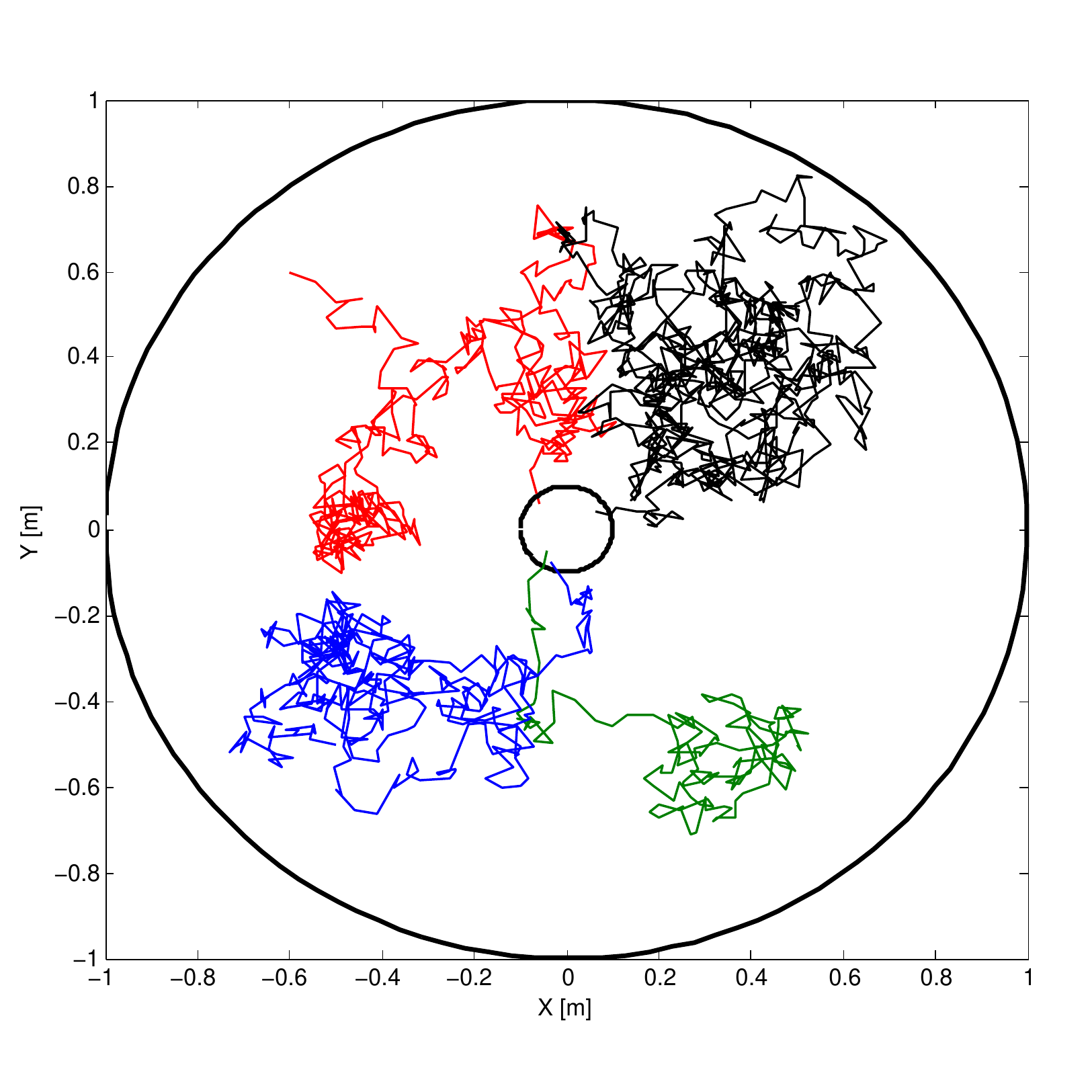}
	\label{fig:intuitive-example-single}
}
\subfigure[Recursive Execution]{
             \includegraphics[width=0.44\textwidth]{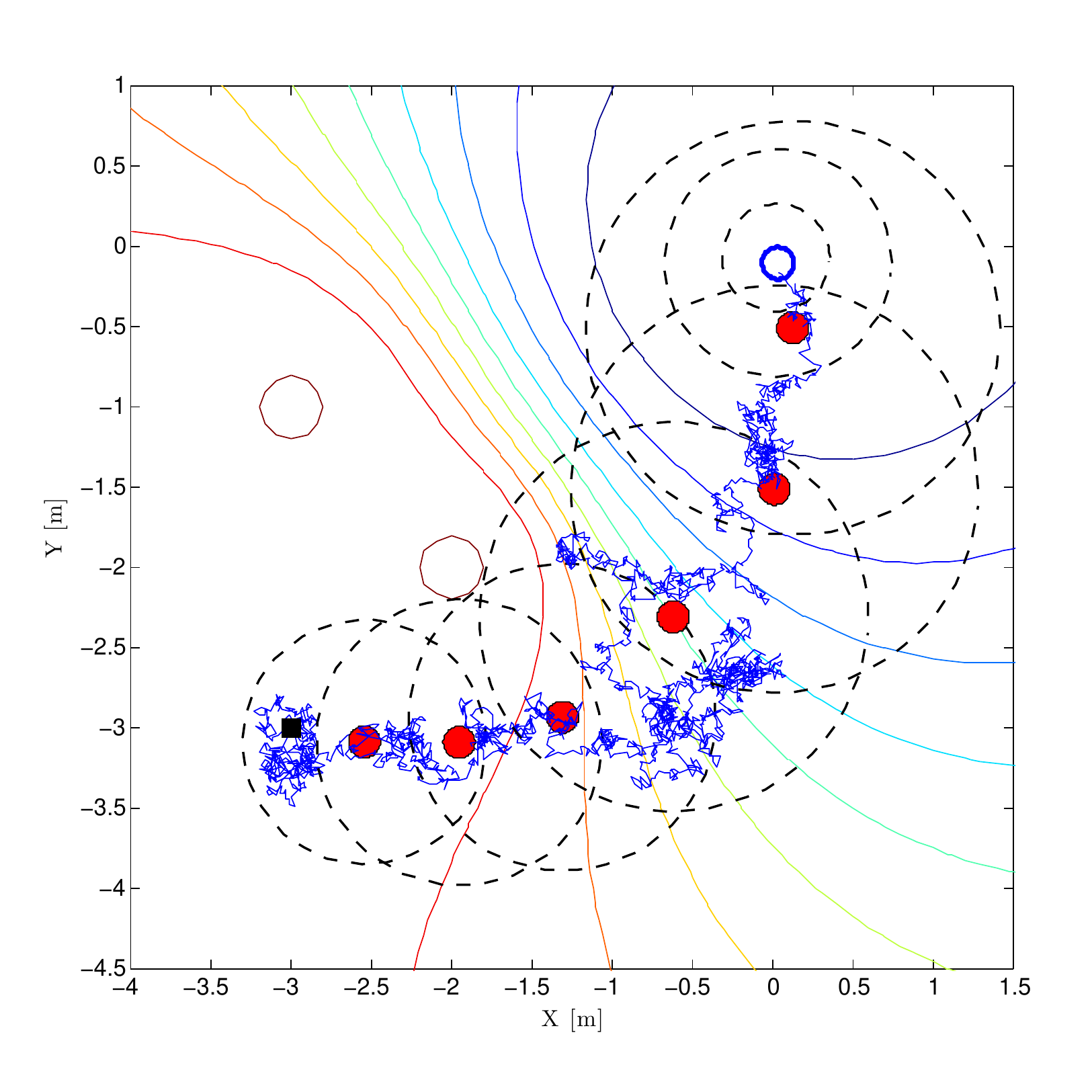}
	\label{fig:intuitive-example-rhc}
}
\caption[]{The stochastic optimal controller \subref{fig:intuitive-example-single}.  Sample paths for a single integrator in 2D. \subref{fig:intuitive-example-rhc} The trajectory resulting from implementation of the stochastic optimal controller in a receding horizon framework.\label{fig:intuitive-example}}
\end{figure}

Assume now that as soon as the system hits the circle of radius $\varepsilon$ around the origin,
a coordinate transformation occurs which shifts the origin to a point within distance $R$ from
its prior location.  Then the same controller can be reapplied to drive the system to a
$\varepsilon$-neighborhood of the new origin.  An iterative scheme based on this idea can be
used to steer the system from point $A$ to point $B$ in a receding horizon manner. 
A sample trajectory resulting from an implementation of such a receding horizon controller is shown in 
Fig.~\ref{fig:intuitive-example-rhc}.  

While the design of the controller \eqref{eq:intuitive-control} that enables convergence
to way-points is simple for the case of the stochastic single integrator of \eqref{eq:examplesystem}, 
is not the case for general stochastic nonlinear systems.  
In following sections, we outline a mathematical framework that allows the computation 
of receding horizon controllers for more complex stochastic nonlinear systems.


\section{Stochastic Optimal Control with Exit Constraints}
\label{section:dayscontroller}
In this section we design stochastic optimal controllers with exit constraints.  
These controllers guarantee convergence to a given set, 
and satisfaction of state constraint, both with probability one.
Consider the stochastic system \eqref{eq:problemsystem}
\begin{multline*}
\dif \mathrm{q}(t) = b\big(\mathrm{q}(t)\big)\dif t + G\big(\mathrm{q}(t)\big) \left[
u\big(\mathrm{q}(t)\big) \dif t \right. \\
\left. + \Sigma\big(\mathrm{q}(t)\big)\dif W(t)\right]\enspace, \quad \mathrm{q}(0) = \mathrm{q}_0
\end{multline*}
which evolves within a bounded domain $\mathcal{D} \subseteq \mathcal{P}$ 
with a $\mathcal{C}^2$ boundary $\partial \mathcal{D}$ and closure denoted $\overline{\mathcal{D}}$. 
Assume that $b(\mathrm{q})$, $G(\mathrm{q})$, $\Sigma(\mathrm{q})$, and $\Sigma^{-1}(\mathrm{q})$ are bounded and Lipschitz continuous on 
$\mathcal{D}$. The objective is to find the control $u(\mathrm{q})$ that yields 
\begin{multline}
\label{Day-COST}
V(\mathrm{q},t) =  \min_{u(\mathrm{q})} \mathbb{E}\left[ 
\int_{0}^{t \wedge \tau_{\mathcal{D}}}L(\mathrm{q}(s), u(s))\dif s  \right. \\ 
\left. +  \Phi\big(\mathrm{q}(t \wedge 
\tau_{\mathcal{D}})\big) \; \Big| \; \mathrm{q}(0) = \mathrm{q} \right],
\end{multline}
where $\tau_{\mathcal{D}}$ is the first exit time from the domain $\mathcal{D}$. (Notation $t \wedge 
\tau_{\mathcal{D}}$ is standard for $\min(t,\tau_{\mathcal{D}})$.) The incremental cost $L(\mathrm{q},u)$
in \eqref{Day-COST} is defined as
\begin{equation*}
L(\mathrm{q}, u) \triangleq 
l(\mathrm{q},t) + \frac{1}{2}u^\intercal\, a^{-1}(\mathrm{q})\, u 
\end{equation*}
where $a(\mathrm{q}) =  \Sigma(\mathrm{q}) \, \Sigma^\intercal(\mathrm{q})$. We impose an admissibility condition that there exist a set of control inputs $u^\mathrm{q} \in U^\mathrm{q}$ such that for all initial conditions $\mathrm{q}$ and control inputs $u^\mathrm{q}$, the cost $V(\mathrm{q},\tau_{\mathcal{D}}) < \infty$.

The \ac{hjb} equation associated with \eqref{Day-COST} is
\begin{equation}
\label{Day-hjb}
\min_{u(\mathrm{q})}\Big\{\mathcal{A} V(\mathrm{q},t) + L\big(\mathrm{q}(t),u(t)\big)\Big\} = 0
\end{equation}
where $\mathcal{A}$ the second-order partial differential operator 
\begin{equation*}\label{generator}
\mathcal{A} \triangleq \frac{\partial}{\partial t} 
+ \sum_{j=1}^n \big(b_j(\mathrm{q}) + G_j(\mathrm{q})u_j(\mathrm{q})\big) \frac{\partial}{\partial \mathrm{q}_j} 
+ \frac{1}{2}\sum_{j=1}^n \sum_{k=1}^n a_{jk}(\mathrm{q}) \frac{\partial^2}{\partial \mathrm{q}_j \partial \mathrm{q}_k} \ .
\end{equation*}
Equation \eqref{Day-hjb} is written in matrix form as follows
\begin{multline*}
\min_{u(\mathrm{q})}\Big\{\partial_t V(\mathrm{q},t) + \partial_\mathrm{q} V^\intercal(\mathrm{q},t)\; b(\mathrm{q}) 
+ \partial_\mathrm{q} V^\intercal(\mathrm{q},t)\; G(\mathrm{q}) \, u(\mathrm{q}) \\
+ \frac{1}{2}\mathrm{tr} \big\{\partial_{\mathrm{qq}}V(\mathrm{q},t) \; G(\mathrm{q})\, \Sigma(\mathrm{q})\, \Sigma^\intercal(\mathrm{q})
G^\intercal(\mathrm{q})\big\} \\
 + l(\mathrm{q},t)+\frac{1}{2}u^\intercal(\mathrm{q})\, a(\mathrm{q})^{-1}\, u(\mathrm{q})\Big\} = 0 
\end{multline*}
where $\mathrm{tr}$ stands for trace.
The optimal control law $u^{\ast} \in U^\mathrm{q}$ that solves \eqref{Day-hjb} is then given as
\begin{equation}
\label{eq:Dayoptimalcontrol}
u^*(\mathrm{q})= - a(\mathrm{q}) G^\intercal(\mathrm{q})\;\partial_\mathrm{q} V(\mathrm{q},t) \ .
\end{equation}
Substituting \eqref{eq:Dayoptimalcontrol} in \eqref{Day-hjb} yields
\begin{multline}
\label{eq:hjb-withcontrol}
 \partial_t V(\mathrm{q},t) + \partial_\mathrm{q} V^\intercal(\mathrm{q},t)\;b(\mathrm{q})  \\
 - \frac{1}{2}\partial_\mathrm{q} V^\intercal(\mathrm{q},t)\;G(\mathrm{q})\, a(\mathrm{q})\, G^\intercal(\mathrm{q})\; \partial_\mathrm{q} V(\mathrm{q},t) \\
+ \frac{1}{2}\mathrm{tr}\big\{\partial_{\mathrm{qq}}V(\mathrm{q},t)\;G(\mathrm{q})\, \Sigma(\mathrm{q})\, \Sigma^\intercal(\mathrm{q})\, 
G^\intercal(\mathrm{q})
\big\} \\ + l(\mathrm{q},t) = 0 \ .
\end{multline}

Using the logarithmic transformation \cite{Fleming}
\begin{equation*}
\label{solution}
V(\mathrm{q},t) = - \log g(\mathrm{q}) \ ,
\end{equation*}
and with substitution in \eqref{eq:hjb-withcontrol} we get
\begin{multline}
\label{eq:PDE-g(x)}
- \partial_t g(\mathrm{q},t) = -l(\mathrm{q},t)\,g(\mathrm{q},t) +  \partial_\mathrm{q} g^\intercal(\mathrm{q},t)\;b(\mathrm{q}) \\
+ \frac{1}{2}\mathrm{tr}\big\{\partial_{\mathrm{qq}}g(\mathrm{q},t)\; G(\mathrm{q})\, \Sigma(\mathrm{q})\, \Sigma^\intercal(\mathrm{q})\, 
G^\intercal(\mathrm{q})\big\} = 0
\end{multline}
with boundary condition
\[g(\mathrm{q}, t \wedge  \tau_{\mathcal{D}}) = \exp\Big(-\Phi\big(\mathrm{q}(t \wedge \tau_{\mathcal{D}})\big)\Big), 
\enspace \mathrm{q} \in \partial \mathcal{D} \enspace .
\]

Analytic solutions of the above \ac{pde} are generally not possible for complex nonlinear systems. 
However, the Feynman-Kac formula \cite{Karatzas} relates a certain \ac{pde} with an equivalent 
\ac{sde}, and facilitates the numerical solution of the \ac{pde} through numerical simulation of the \ac{sde}.
Using the Feynman-Kac formula \cite{Karatzas}, the solution of \eqref{eq:PDE-g(x)} takes the form
\begin{align}
\label{eq:feynmen-kac}
\nonumber g(\mathrm{q}) &= \mathbb{E} \left[ g(\mathrm{q}, t \wedge \tau_{\mathcal{D}}) \; 
\exp \left( \int_0^{t \wedge \tau_{\mathcal{D}}} l(\mathrm{q},s)\dif s \right) \;\Big|\; \zeta(0) = \mathrm{q} \right] \\
&= \mathbb{E} \left[ \exp\Big( -\Phi\big(\zeta(t \wedge \tau_{\mathcal{D}})\big)\Big)
\exp \left( \int_0^{t \wedge \tau_{\mathcal{D}}} l(\mathrm{q},s)\dif s \right)  \;\Big|\; \zeta(0) = \mathrm{q} \right]
\end{align}
where $\zeta(t)$ is the Markov process 
\begin{equation}
\label{eq:aux-process}
\dif \zeta(t) = b\big(\zeta(t)\big)\dif t + G\big(\mathrm{q}(t)\big)\, \Sigma\big(\zeta(t)\big) \dif W(t) 
\end{equation}
evolving on the same bounded open set $\mathcal{D} \subset \mathbb{R}^n$.

\paragraph*{Stochastic Optimal Control with Exit Constraints}

Under the assumption 
\begin{equation}
\label{eq:finite-time-condition}
\min_{\mathrm{q} \in \overline{\mathcal{D}}} a_{ll}(\mathrm{q}) > 0
\end{equation}•
for some $1 \leq l \leq m$, one 
can show that $\mathbb{E} [\tau_{\mathcal{D}} \mid \mathrm{q}(0) = \mathrm{q}_0] < \infty$, $\forall \mathrm{q}_0 \in 
\overline{\mathcal{D}}$ \cite[Lemma 7.4]{Karatzas}. This means that the system will escape 
the domain $\mathcal{D}$ in finite time with probability one. The assumption that $\Sigma$ and $\Sigma^{-1}$ are bounded, ensures satisfaction of \eqref{eq:finite-time-condition}.

A guarantee that the system does not exit from a specific portion of the boundary can be obtained
by imposing an infinite penalty for touching that surface.  
Consider a partition of the boundary $\partial \mathcal{D}$ in the form 
$\mathcal{N} \subset \partial \mathcal{D}; \enspace \mathcal{M} 
= \partial \mathcal{D} \setminus \mathcal{N}$.  
Then choose $\Phi$ as  
\begin{equation*} 
\label{Phi}
\Phi=+\infty \cdot \mathcal{X}_{\mathcal{M}};
\end{equation*}
and
\begin{align*}
\mathcal{X}_{\mathcal{M}}  &= 
\begin{cases} 0 & \text{on } {\mathcal{N}}\\
1 &\text{on } {\mathcal{M}} 
\end{cases} 
\end{align*}
Assuming that $l(\mathrm{q},t) \equiv 0$ and letting $t \to \infty$, 
the resulting parabolic \ac{pde} \eqref{eq:PDE-g(x)} gives rise to the Dirichlet problem
\begin{align}
\label{eq:Dirichlet}
\partial_\mathrm{q} g^\intercal(\mathrm{q})\, b(\mathrm{q}) &+ 
\frac{1}{2}\mathrm{tr}\big\{\partial_{\mathrm{qq}}g(\mathrm{q})\;G(\mathrm{q})\, \Sigma(\mathrm{q})\,  \Sigma^\intercal(\mathrm{q}) \, G^\intercal(\mathrm{q})\big\} = 0 \\
\nonumber & \begin{cases} g\big(\mathrm{q}(\tau_{\mathcal{D}})\big) = 1 
&  \mathrm{q}(\tau_{\mathcal{D}}) \in \mathcal{N} \\
g\big(\mathrm{q}(\tau_{\mathcal{D}})\big) = 0 & \mathrm{q}(\tau_{\mathcal{D}}) \in \mathcal{M}
\end{cases}
\end{align}
Then \eqref{eq:feynmen-kac} suggests that $g(\mathrm{q})$ is in fact the probability that 
the sample path of \eqref{eq:aux-process} 
from $\mathrm{q}$ hits boundary $\mathcal{N}$ before $\mathcal{M}$.
Function $g(\mathrm{q})$ takes the form
\begin{equation}
\label{eq:g(x)}
g(\mathrm{q}) = \mathbb{P} \left[\zeta(\tau_{\mathcal{D}}) \in {\mathcal{N}} \mid \zeta(0) = \mathrm{q}\right].
\end{equation}
and $\zeta(t)$ is the Markov process \eqref{eq:aux-process}. Now if the admissibility condition is satisfied then the optimal control with infinite penalty on exit boundary is equivalent to a constraint (see \cite{Day}), 
\[\mathbb{P}\big[ \mathrm{q}(\tau_{\mathcal{D}}) \in  \mathcal{M} 
\mid \mathrm{q}(0) = \mathrm{q}\big] = 0 
\]

\begin{remark}
The computation of control input \eqref{eq:Dayoptimalcontrol} requires $g(\mathrm{q})$, which can be found by either by solving \eqref{eq:PDE-g(x)} analytically, or numerically simulating \eqref{eq:aux-process} and computing \eqref{eq:g(x)}. As $\Phi$ imposes an infinite penalty on state trajectories that exit through $\mathcal{M}$, the above construction forces the system to exit through $\mathcal{N}$ while avoiding $\mathcal{M}$ with probability one.  The problem of stochastic optimal control with terminal cost at exit time is discussed in \cite{Fleming}, while a specific problem of exit constraints was discussed in \cite{Day}. The latter reference also shows that imposing an exit constraint is equivalent to having infinite penalty on exit location used in this section. We use these two results and thus by defining $\mathcal{M}$ to be the boundary of state constraint regions, we achieve the guarantees that state constraints are satisfied, and convergence to a desired region is achieved in finite time. 
\end{remark}

\section{Stochastic Receding Horizon Control Design}
\label{section:framework}

After the presentation of the continuous-time constrained stochastic optimal control
formulation in its general setting, we proceed with the description of the implementation
of these techniques inside the receding horizon framework that was outlined in the
example of Section~\ref{section:intuitiveexample}.  Out of this process emerges a simple,
special case of a \ac{gshs}, for which the existence of solutions has been established in
literature~\cite{bujorianu-lygeros}.  The section concludes with an examination of the
closed loop stability and convergence properties of this simplified \ac{gshs}, and a discussion
on how input saturation affects these properties.

\subsection{Deterministic Planning}
We begin by computing a receding horizon path using \eqref{deterministic-dynamics}
\begin{equation*}
\dot{\hat{\mathrm{q}}} = b(\hat{\mathrm{q}}) + G(\hat{\mathrm{q}})u(\hat{\mathrm{q}}) \ .
\end{equation*}
Let $\hat{\mathrm{q}}^*_T(t):[t_0,t_0+T] \to \mathbb{R}^n$ be the trajectory that,
for a prediction horizon $T$, minimizes the cost functional
\begin{align*}
\label{eq:discretecost}
 J_T(\mathrm{q},u) =& \min_{\hat{u}(t)}\int_{t_0}^{t_0 + T} L\big(\hat{\mathrm{q}}(s),u(s)\big)\dif s + Q\big(\hat{\mathrm{q}}(t_0+T)\big) \\ 
&\text{subject to} \enspace
\inf_{\substack{z \in \mathcal{O} \\ t\in[t_0,t_0+T]}} \| \hat{\mathrm{q}}(t) - z\| > R \ , \qquad  \hat{\mathrm{q}}(t_0) = \mathrm{q} \ 
\end{align*}
with functions $L$ and $Q$ as in \eqref{standard-rh}.
Define a \emph{receding horizon path} as
\begin{equation}
\label{eq:path}
\Gamma_T \triangleq \{ \gamma \in \mathbb{R}^n \mid \exists t \in [t_0,t_0+T] : \gamma = \hat{\mathrm{q}}^*_T(t)\} \ .
\end{equation}

Here we adopt the approach of \cite{Tanner-10} to obtain an approximation of $\hat{\mathrm{q}}^\ast_T$ and
consequently compute $\Gamma_T$.  The latter, however, can also be obtained through an array of alternative
methodologies, including potential field methods \cite{Kodi2}, rapidly exploring random trees
\textsc{rrt}s \cite{Lavalle00rapidly-exploringrandom}, or cell decomposition methods \cite{Lav06}.


\subsection{Way-point Generation}
\label{way-points}
Let the closed ball of radius $\varepsilon$ centered at a point $\gamma$ is denoted $\overline{\mathcal{B}}_\gamma(\varepsilon) \triangleq \{ \mathrm{q} :  \| \mathrm{q} - \gamma \| \leq \varepsilon \}$,
and its complement, $\mathcal{B}^c_\gamma(\varepsilon)$.
Now consider a sequence of points
$\{ \gamma_i\}_{i = 0}^{N} \in \Gamma_T$ with $\gamma_0 := \mathrm{q}(t_0)$ and $\gamma_N := \hat{\mathrm{q}}^*_T(t_0 + T)$, satisfying
\begin{equation}
\label{eq:convergence-condition}
\max_{a \in \overline{\mathcal{B}}_{\gamma_i}(\varepsilon)}\{Q(a)\} 
- \min_{b \in  \overline{\mathcal{B}}_{\gamma_{i-1}}(\varepsilon)}\{Q(b)\} \leq -\eta(\|\gamma_{i-1}\|)
\enspace,
\end{equation}
where $\gamma$ is the positive definite function in \eqref{condition-of-ali}.
Define domains $\mathcal{D}_i$, for $i = 1,\ldots,N$, such that
$ \bigcup_i \mathcal{D}_i \cap \mathcal{O} = \emptyset$ and
\begin{equation}
\label{eq:conditions}
 \overline{\mathcal{B}}_{\gamma_{i-1}}(\varepsilon) 
 \subset \mathcal{D}_i \subset \mathcal{B}^c_{\gamma_i}({\varepsilon}) 
\end{equation}
Decompose the boundaries of those domains as follows (see Fig.~\ref{example-localdomains}):
\begin{align}
\mathcal{N}_i &\triangleq \partial \mathcal{D}_i \cap \overline{\mathcal{B}}_{\gamma_i}({\varepsilon})  \\
\label{eq:boundary-definition}
\mathcal{M}_i &\triangleq \partial \mathcal{D}_i \setminus \mathcal{N}_i
\end{align}
The domains $\mathcal{D}_i$ are defined such that $\mathcal{N}_i$ is non-empty for all $i$.

\begin{figure}[h]
      \centering
      \includegraphics[width=0.45\textwidth]{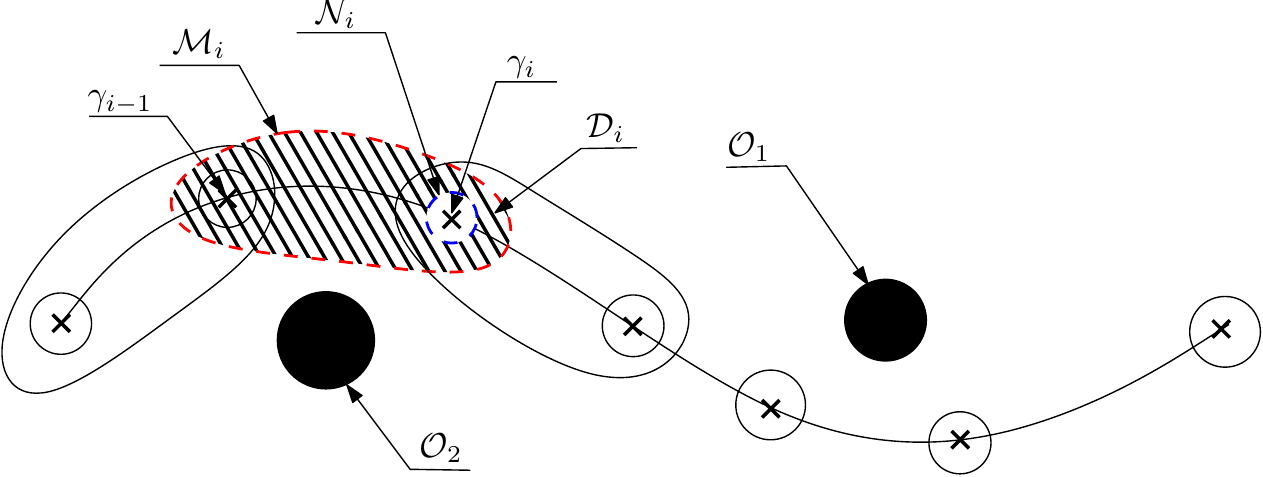}
      \caption{Illustration of the local domains $\mathcal{D}_i$ (hashed region). 
      Also shown are the receding horizon path $\Gamma_T$ (continuous curve),
      the way-points defined by the sequence $\{\gamma_i\}$ (crosses),
      the obstacles $\mathcal{O}_j$ (solid disks), the boundaries $\mathcal{N}_i$ 
      (dashed blue inner boundary) 
      and $\mathcal{M}_i$ (dashed red outer boundary).   \label{example-localdomains} }
   \end{figure}

\subsection{Stochastic optimal controllers}

The system state is a Markov process $\mathrm{q}(t)$ that 
evolves between way-points according to the \ac{sde}
\begin{equation}
\dif \mathrm{q}(t) = b\big(\mathrm{q}(t)\big)\dif t + G\big(\mathrm{q}(t)\big) \big[u_i\big(\mathrm{q}(t)\big)\dif t 
+ \Sigma\big(\mathrm{q}(t)\big) \dif W(t)\big]
\label{localSDE1}
\end{equation}
where $\Sigma(\mathrm{q})$, $b(\mathrm{q})$, $G(\mathrm{q})$, $\Sigma^{-1}(\mathrm{q})$ satisfy
the requirements of Section~\ref{section:dayscontroller}, and together with
$u_i$, are all bounded in $\mathcal{D}_i$.  The latter 
is the control input responsible for taking the state
from $\mathcal{N}_{i-1}$ to $\mathcal{N}_i$ while
avoiding $\mathcal{M}_i$.   
Let $t_{i-1}$ be the first time instant when
$\mathrm{q}(t) \in \mathcal{N}_{i-1}$.  

When \eqref{localSDE1} under $u_i$ hits $\mathcal{N}_i$ at some time $t_i$, 
it undergoes a forced transition with $u_i$ switching to $u_{i+1}$, and the switch occurs upon the state hitting a part of the boundary $\mathcal{N}_i$. 
Control law $u_i$ gives a solution to the stochastic optimal control problem 
\begin{multline}
\min_{u_i} \mathbb{E} \left[ \int_{t_{i-1}}^{t_{i}} \frac{1}{2}
u_i^\intercal\big(\mathrm{q}(s)\big)\,a^{-1}\big(\mathrm{q}(s)\big)\, u_i\big(\mathrm{q}(s)\big) \dif s \right.\\
\left. + \Phi\big(\mathrm{q}(t_{i})\big) \;\Big|\; \mathrm{q}(t_{i-1}) = \mathrm{q} \right] =: V(\mathrm{q})
\label{localcost1}
\end{multline}
Notice that by setting now the terminal time to $t_i$ allows the value function $V$ to be
time-invariant.  We define the exit time for the process driven by $u_i$ to be 
$\tau_i = t_i - t_{i-1}$.
Function $\Phi$ is again chosen in a way that it imposes infinite on the
state hitting $\mathcal{M}_i$.  Similarly to the analysis of Section \ref{section:dayscontroller},
the solution of \eqref{localcost1} is
\begin{equation*}
V(\mathrm{q}) = -\log g(\mathrm{q})
\end{equation*}
where $ g(\mathrm{q}) = \mathbb{P} \left[\zeta(\tau_{i}) \in {\mathcal{N}_i} \mid \mathrm{q}(t_{i-1}) = \mathrm{q}\right]$,
and the optimal control law for $\mathrm{q} \in \mathcal{D}_i$ is
\begin{equation}
\label{eq:optimalcontrol}
{u_i}^*(\mathrm{q}) =  - a(\mathrm{q})\,G^\intercal(\mathrm{q})\;\partial_\mathrm{q} V(\mathrm{q}) \ .
\end{equation}

When applied, $u^*_i(\mathrm{q})$ satisfies the following probabilistic conditions:  
\begin{align}
\label{eq:finiteexittime} &\mathbb{E}\big[\tau_{i} \mid \mathrm{q}(t_{i-1}) = \mathrm{q}\big] 
< \infty\\
\label{eq:probabilitycondition}& \mathbb{P}\big[ \mathrm{q}(t_{i}) \in  \mathcal{M}_i 
\nonumber \mid \mathrm{q}(t_{i-1}) = \mathrm{q}\big] = 0 \\
 & \iff \mathbb{P}\big[\mathrm{q}(t_{i}) \in  \mathcal{N}_i \mid \mathrm{q}(t_{i-1}) = \mathrm{q}\big] = 1\ .
\end{align}
Condition \eqref{eq:finiteexittime} translates into the process $\mathrm{q}(t)$ exiting 
$\mathcal{D}_i$ in finite time with probability one which is guaranteed by assumption \eqref{eq:finite-time-condition}. 
Condition \eqref{eq:probabilitycondition} is equivalent to saying that the process $\mathrm{q}(t)$ 
reaches an $\varepsilon$-neighborhood of way-point $\gamma_i$ with probability one, 
before violating any state constraints (see \cite{Day}).

Given a receding horizon path $\Gamma_T$ seeded with a sequence of way-points $\{\gamma_i\}_{i=0}^N$,
the process of transitioning from way-point $\gamma_{i-1}$ to way-point $\gamma_i$ under \eqref{eq:optimalcontrol}
is repeated.  By the time a new way-point is reached, the path $\Gamma_T$ has been
recomputed in a receding horizon manner, and the way-point sequence $\{\gamma_i\}_{i=0}^N$ 
redefined with the initial element $\gamma_0$ being the way-point just reached.
What is important for real-time implementation is that for predetermined domains $\mathcal{D}_i$,
\eqref{eq:optimalcontrol} can be precomputed off-line, numerically in general but
also analytically in special cases where $b$, $G$ and $\Sigma$ are such that the boundary 
value problem for \ac{pde} \eqref{eq:PDE-g(x)} can be solved explicitly.

\subsection{The Resulting Stochastic Hybrid System}

Closing the loop around \eqref{localSDE1} by means of a receding horizon strategy 
gives rise to a switched stochastic hybrid system, where
switching is due to $u_i$ and occurs as a forced transition whenever $\mathrm{q}(t)$ hits
a set $\mathcal{N}_i$. 
The hybrid state here is just $(i, \mathrm{q})$ where $\mathrm{q} \in \mathbb{R}^n$ and 
$i \in \{0,1,2, ..., N\} =: \mathcal{I}$ are the continuous and discrete states, respectively.
This system can be classified as a \ac{gshs}, a general modeling framework of which is 
described in \cite{bujorianu-lygeros}; however, it is a very simplified version of the
the general definition of \cite{bujorianu-lygeros}, which can be adequately described
by defining only the following three components: the continuous dynamics, the
discrete dynamics, and the reset condition.

\paragraph*{Continuous Dynamics} The continuous state $\mathrm{q}(t)$ evolves according to the \ac{sde} 
\eqref{localSDE1}
\begin{equation}
\label{eq:SHS-Continuous}
\dif \mathrm{q}(t) = b\big(\mathrm{q}(t)\big) \dif t + G\big(\mathrm{q}(t)\big) \left[u\big(i,\mathrm{q}(t)\big)
\dif t + \Sigma\big(\mathrm{q}(t)\big)\dif W(t)\right]
\end{equation}
where we have just replaced $u_i\big(\mathrm{q}(t)\big)$ with $u\big(i,\mathrm{q}(t))$ to emphasize the
explicit dependence of the control input on the discrete state $i$, making it a function
of the hybrid state $(i,\mathrm{q})$:
$u:\mathcal{I} \times \mathbb{R}^n \to \mathbb{R}^m$.
The drift $b$ and diffusion $\Sigma$ terms, along with $G$, are assumed independent of $i$.
When in discrete state $i$, the domain of the continuous variable $\mathrm{q}(t)$ is $\mathcal{D}_i$.

\paragraph*{Discrete Dynamics} The (single) discrete state $i$ evolves by means of
state-triggered
forced transitions, which occur each time the continuous state $\mathrm{q}$ hits a \emph{guard}.
In this case the guard is a function from $i$ to $\mathbb{R}^n$, sending $i \mapsto \mathcal{N}_i$.
The time at which the transition is triggered is called \emph{stopping time} and it is the
first time instant $t_i \triangleq \inf\{t > t_{i-1} \mid \mathrm{q}(t) \notin \mathcal{D}_i \}$.  Then the
discrete state changes according to the following---in fact, deterministic---rule:
\begin{align*}
\mathbb{P}\big(i + 1 \mid i, \mathrm{q}(t_{i-1}) = \mathrm{q}\big) = 
\begin{cases}
 1  &  \mathrm{q}(t_i) \in \mathcal{N}_i \\
   0  & \text{otherwise}  \enspace.
 \end{cases}
\end{align*}
Note that due to the set of discrete states being finite, and the discrete transition map
being a bijection, there can only be a finite number of discrete transitions and the system
cannot exhibit Zeno behavior.

\paragraph*{Reset Condition} During discrete transitions, continuous states are not reset.
Essentially, the reset map for the continuous states is simply the identity.

The solution of \eqref{eq:SHS-Continuous} over $i=1,\ldots,N$, 
is a collection of Markov processes truncated at (their) exit time, which
can be represented as a \emph{Markov string}.  
A Markov string is a hybrid state jump Markov process \cite{bujorianu-lygeros}. 
Given the existence of solutions for each \ac{sde} \eqref{localSDE1} for fixed $i$ (see \cite{Day}
for details), and due to the finiteness of the set of discrete states, the solutions for the closed
loop stochastic hybrid system are well defined \cite{bujorianu-lygeros}.



\section{Convergence and Stability Properties}
\label{section:stability}

This section presents a proposition that establishes the 
finite-time convergence properties of the closed
loop system to a neighborhood of the origin.

\begin{proposition}
Consider the switched stochastic system \eqref{localSDE1} in an open bounded domain $\mathcal{S} \subset 
\mathbb{R}^n$, where $i \in \mathcal{I}$ is the switching index, and $W(t)$ is a Wiener process. 
Let $Q(\mathrm{q})$ be a $\mathcal{C}^2$, positive definite function in the closure of a bounded domain 
$\mathcal{S}$ which contains the origin. If for every solution $\mathrm{q}(t)$ of the stochastic switched system there exist
\begin{compactenum}[(i)]
\item bounded domains $\mathcal{D}_i$ that satisfy \eqref{eq:conditions}--\eqref{eq:boundary-definition}, and
\item a class-$\mathcal{K}$ function $\eta$ on $\mathcal{S}$ together with
a sequence of points $\{\gamma_i\}^{N}_{i=0} \in \mathcal{S}$ satisfying \eqref{eq:convergence-condition},
\end{compactenum}
then 
the closed-loop switched stochastic system \eqref{localSDE1}--\eqref{eq:optimalcontrol} converges to an 
$\varepsilon$-neighborhood of origin in finite time.
\end{proposition}

\begin{proof}
It is known \cite{jad-01} that a receding horizon strategy $u_\mathsf{rh}(t)$ applied on \eqref{standard-rh} yields a trajectory $\hat{\mathrm{q}}^*(t)$ satisfying $\lim_{t\to\infty} \hat{\mathrm{q}}^*(t) \to 0$. Hence, with sufficiently large $T<\infty$, one can find a path $\Gamma_T$ such that $\Gamma_T \cap \mathcal{B}_0(\varepsilon) \neq \emptyset$. Moreover, condition \eqref{condition-of-ali} ensures that for any $\mathrm{q}(t_0) \in \mathcal{S}$, the system will remain within an open bounded set containing the level set of $\mathrm{q}(t_0)$. This means that for a sufficiently large $T$, the path $\Gamma_T$ intersects an 
$\varepsilon$-neighborhood of the origin and remains bounded.
Given that this set is bounded,
one can only cover it with a finite number of non-overlapping balls with radius $\varepsilon > 0$. Hence,  for sufficiently large $T< \infty$, there is a finite number of way-points $N$ that satisfy condition \eqref{eq:convergence-condition} with $\gamma_N$ at the origin. Then, by induction it is shown 
in a straightforward way that the system reaches an $\varepsilon$-neighborhood of the origin in finite time.

%

To this end, set $\mathrm{q}(t_0) = \gamma_0$, construct a path $\Gamma_T$ of finite length according to \eqref{standard-rh}, and select a way-point $\gamma_1$ according to \eqref{eq:convergence-condition}. Given that bounded domain $\mathcal{D}_1$ satisfies \eqref{eq:conditions}--\eqref{eq:boundary-definition},
the application of control law \eqref{eq:optimalcontrol} ensures that
for all $\mathrm{q}(t_0) \in \mathcal{N}_{0}$, $\mathbb{P}\{\mathrm{q}(t_1) \in \mathcal{N}_1 \mid \mathrm{q}(t_0) \} = 1$, that is, the state at time $t_1$ is in $\mathcal{N}_1$ almost surely (see Section \ref{section:dayscontroller} and \cite{Day}). Condition \eqref{eq:finite-time-condition} ensures that the time that this happens is finite.

Now, let us assume that a controller $u_k(\mathrm{q})$ was applied iteratively, and at some time 
$t_k$,  state
$\mathrm{q}(t_k) \in \mathcal{N}_k$.
As $\mathcal{N}_k \subset \mathcal{D}_{k+1}$ and given \eqref{eq:conditions}, there exists
a controller $u_{k+1}(\mathrm{q})$ to steer the state to the next way-point $\gamma_{k+1}$.
Given now that $\mathcal{D}_{k+1}$ also satisfies \eqref{eq:conditions}--\eqref{eq:boundary-definition}, the law \eqref{eq:optimalcontrol} gives 
$\mathbb{P}\{\mathrm{q}(t_{k+1}) \in \mathcal{N}_{k+1} \mid \mathrm{q}(t_{k})\} = 1$
with  $\mathbb{E}[t_{k+1}] < \infty$.  Inductively, since 
 $\mathcal{N}_N := \partial \mathcal{D}_N \cap \overline{\mathcal{B}}_{0}(\varepsilon)$,
 the proof is completed.
\end{proof}

\subsection{Convergence under bounded inputs}

The control law $u_i(\mathrm{q}) = - a(\mathrm{q}) \, G^\intercal(\mathrm{q})\, 
\partial_\mathrm{q} \{-\log g(\mathrm{q})\}$ may require large inputs 
near the boundary $\mathcal{M}_i$, since $g(\mathrm{q}) \to 0$ there.  
This can be problematic from an implementation standpoint.
When these inputs saturate at some $\| u(\mathrm{q}) \|_\mathsf{max}$, 
the control law that is practically implemented is rather approximated smoothly by
\begin{equation*}
\label{boundedinputs}
\check{u}_i(\mathrm{q}) = - \| u(\mathrm{q}) \|_\mathsf{max} \cdot \tanh \!\left( 
a(\mathrm{q}) \, G^\intercal(\mathrm{q})\, \partial_\mathrm{q} V(\mathrm{q},t) \right) \enspace.
\end{equation*}
The problem is that
bounded inputs cannot force exit at $\mathcal{N}_i$ with probability one. The probability of success in
exiting when bounded inputs are applied can be computed \cite{shah-iros2011}, but there there is always a nonzero probability that the system will exit from $\mathcal{M}_i$ instead of $\mathcal{N}_i$. Neither convergence to origin nor constraint satisfaction can be guaranteed almost surely.

To recover convergence under bounded inputs, we propose a recovery strategy that uses repeatedly a 
controller precomputed offline, which steers the system back inside the domain $\mathcal{D}_i$. The receding horizon control can be re-initiated after the state is re-enters $\mathcal{D}_i$. This 
recovery controller is not different from \eqref{eq:optimalcontrol}, and its use is illustrated
in an example in Section \ref{section:examples}.
In the absence of obstacles, and with infinitely large outer domain, the guarantee of convergence can thus be recovered even with bounded inputs.

\section{Examples}
\label{section:examples}
We present two different examples to demonstrate application of our control design. 
In the first example the stochastic optimal control law can be computed explicitly, and 
simulation results are presented to demonstrate its function. 
The effect of input saturation is also investigated.   
The second example involves a nonlinear system, where the stochastic optimal control laws can not be 
computed explicitly.  There, we show how the application of the Feynman-Kac formula offers numerical
controller designs, and we present the results through representative plots.
 
\subsection{The Stochastic Single Integrator}
\label{section:linearexample}
\paragraph*{Problem formulation} 
Consider the system \eqref{localSDE1} with the drift term $b(\mathrm{q}) \equiv 0$ and $G(\mathrm{q})$ is identity. 
This simple drift-less system can be described as a two-dimensional single integrator with stochastic uncertainty as
\begin{equation}
\label{eq:linearexample}
\dif \mathrm{q}(t) = u_i(\mathrm{q}(t))\dif t + \Sigma(\mathrm{q}(t))\dif W(t); \enspace \mathrm{q}(0) = \mathrm{q}_0
\end{equation}
where $\mathrm{q} = [x \ y]^\intercal$ is the state and $W(t)$ is a 2-dimensional Wiener process.  The objective is to find control inputs $u_i(\mathrm{q}(t))$ to drive the system to origin,  using minimal inputs, 
avoiding obstacles, and moving along paths of minimal length to its destination. Here the system's workspace is  a ball of radius $\rho_0$, containing $M$ spherical obstacles with radii $\rho_j$ and centers $\mathrm{q}_j$, $j=1,2,\dots,M$. 
\paragraph*{Deterministic Path Planning} 
The first step is to find a reference trajectory for \eqref{eq:linearexample} ignoring noise. The nominal dynamics is just $\dot{\hat{\mathrm{q}}} = u(\hat{\mathrm{q}}(t))$. 
We use the approach of \cite{Tanner-10} (other methods are also possible) 
to find a continuous trajectory minimizing a finite-horizon cost
\begin{equation*}
J(\hat{\mathrm{q}},u) = \int_{0}^{T} \{c_1 \|u(s)\|^2 + c_2 \| \hat{\mathrm{q}}(s)\|^2 \} \dif s + Q(\hat{\mathrm{q}}(T))
\end{equation*}
where $T$ is the prediction horizon and $c_1$ and $c_2$ are arbitrary positive constants. The terminal cost $Q(\hat{\mathrm{q}}(T))$ is selected as a \emph{navigation function} \cite{Kodi} defined as 
\begin{equation}
Q(\mathrm{q}) = \left( \frac{\|\mathrm{q}\|^{2k}}{\|\mathrm{q}\|^{2k} + \beta(\mathrm{q})}\right)^\frac{1}{k}
\label{NavFun}
\end{equation}
where $k \in \mathbb{N}^+$ is a sufficiently large positive integer. In \eqref{NavFun}, the function
$\beta:\mathcal{P} \rightarrow \left[0,\infty\right)$ encodes the location and size of obstacles and is expressed as
\[
\beta \triangleq \prod^{M}_{j=0}\beta_j
\] 
with
$\beta_0  \triangleq \rho^2_0 - \left\|\mathrm{q}\right\|^2$ and  
$\beta_j \triangleq \left\|\mathrm{q} - \mathrm{q}_j\right\|^2 - \rho^2_j $, for  $j = 1, \dots, M$.


Assume that the outcome of this procedure is an obstacle-free continuous state trajectory $\hat{\mathrm{q}}^*(t) \in \mathcal{P}$, and the resulting path is
$
\Gamma \triangleq \{ \gamma \in \mathbb{R}^2 \mid \exists t \in \mathbb{R} ; \gamma = \hat{\mathrm{q}}^*(t)\} \ .
$
\paragraph*{Way-point Generation} 
There exist control way-points $\{\gamma_i\}_{i=0}^{N} \in \Gamma$, such that 
$\gamma_0 = \hat{\mathrm{q}}(t_0)$, and $\gamma_N = \hat{\mathrm{q}}^*(T)$.  Define the sets $\overline{\mathcal{B}}_{\gamma_i}({\varepsilon}) \triangleq \{\mathrm{q} \in \mathcal{P} : \| \mathrm{q} - \gamma_i\| \leq \varepsilon\}$ 
and denote their boundary $\partial \mathcal{B}_{\gamma_i}({\varepsilon})$.
The waypoints we select are chosen to satisfy the following constraint:
\begin{align}
\label{condition1}
&\max_{a \in \overline{\mathcal{B}}_{\gamma_{i-1}}({\varepsilon})}\{Q(a)\} - \min_{b \in \overline{\mathcal{B}}_{\gamma_i}({\varepsilon})}\{Q(b)\} \leq -\eta(\|\gamma_{i-1}\|) \\
\label{condition2}
&\|\gamma_{i-1} - \gamma_i \| > 2\varepsilon \\
\label{condition3}
&R_i < \min\{ \| \gamma_i - z \|, z \in \mathcal{O}\}, \ R_i - 2\varepsilon > \|\gamma_{i-1} - \gamma_i \| 
\end{align}
where $\varepsilon$ and $R_i$ are positive constants. The above constraints also help determine the radius $R_i$, which is the outer radius of the domain of the continuous state $\mathcal{D}_i$. There is no unique solution for $R_i$ and one can specify an upper and lower bounds on $R_i$. 

The local domains $\mathcal{D}_i$ are now defined as
\begin{align*}
& \mathcal{D}_i \triangleq \mathcal{B}_{\gamma_i}({R_i}) \setminus \overline{\mathcal{B}}_{\gamma_i}({\varepsilon})\\
&\partial \mathcal{D}_i \triangleq \partial \mathcal{B}_{\gamma_i}({\varepsilon}) \cup \partial \mathcal{B}_{\gamma_i}({R_i}) \ .
\end{align*}
where $\mathcal{N}_i = \partial \mathcal{B}_{\gamma_i}({\varepsilon}) $ and $\mathcal{M}_i = \partial \mathcal{B}_{\gamma_i}({R_i})$. 
Conditions \eqref{condition1}--\eqref{condition2} imply that
\[\mathcal{B}_{\gamma_i}({R_i}) \cap \mathcal{O} = \emptyset; \enspace \mathcal{B}_{\gamma_{i-1}}({\varepsilon)} \subset \mathcal{D}_i \ ,  \forall i \ . \]
\paragraph*{Stochastic optimal controller} 
The control input $u_i(\mathrm{q}(t_i))$ for \eqref{eq:linearexample} is constructed as
shown in Section \ref{section:dayscontroller}.  It achieves 
\begin{equation*}
V(\mathrm{q}) = \min \mathbb{E} \left[ \frac{1}{2} \int_{t_{i-1}}^{t_{i}} u(\mathrm{q}(s))^\intercal u(\mathrm{q}(s))\dif s + \Phi(\mathrm{q}(t_{i})) \:\Big|\: \mathrm{q}(t_{i-1}) =\mathrm{q}  \right]
\end{equation*}
where 
\begin{equation*} 
\Phi=+\infty \cdot \mathcal{X}_{\mathcal{M}_i}; \ 
\mathcal{X}_{\mathcal{M}_i}  = 
\begin{cases} 0 & \text{on } {\mathcal{N}_i}\\
1 &\text{on } {\mathcal{M}_i} .
\end{cases} 
\end{equation*}
The optimal control law is  
\[u^*(\mathrm{q}) =  - a(\mathrm{q})\cdot\partial_\mathrm{q} V(\mathrm{q})\]
where $a(\mathrm{q}) = \Sigma(\mathrm{q})\Sigma^\intercal(\mathrm{q})$,
$V(\mathrm{q}) = - \log g(\mathrm{q})$, 
and $g(\mathrm{q})$ is the solution of the \ac{pde}
\begin{align*}
\frac{1}{2}\left(\frac{\partial^2}{\partial x^2}+ \frac{\partial^2}{\partial y^2}\right) g &= 0& & 
\text{in } \mathcal{D}_i\\
g &= 0 & &\text{on } \mathcal{M}_i \\ g  &= 1& &\text{on } \mathcal{N}_i
\end{align*}
Function $g(\mathrm{q})$ has an analytic expression:
\begin{equation*}
g(\mathrm{q}) = \frac{R_i- {\| \mathrm{q} -\gamma_i \|}}{R_i - \varepsilon} \enspace,
\end{equation*}
which suggests a value function
\[V(\mathrm{q}) = - \log \frac{R_i - {\| \mathrm{q} -\gamma_i \|}}{R_i - \varepsilon} \]
and a control law of the form
\begin{equation}
\label{actualcontrol}
u_i(\mathrm{q}) = -a(\mathrm{q})\cdot 
\frac{\mathrm{q} -\gamma_i }{{\big(R_i -\| \mathrm{q} -\gamma_i \|\big)}
{\| \mathrm{q} -\gamma_i \|}} \enspace.
\end{equation}
Control input $u_i(\mathrm{q})$ switches to $u_{i+1}(\mathrm{q})$ upon hitting the boundary $\mathcal{N}_i$ for $i = 1, 2, \dots$ until the state is in $\varepsilon$-neighborhood of the goal.

\paragraph*{Problem instantiation and simulation results}
Simulations were performed (taking $\mathrm{q} \in \mathbb{R}^2$) with the overall bounded domain being $\mathcal{S} = \{\mathrm{q} \in \mathbb{R}^2 \mid \| \mathrm{q} \| < 10 \}$.  
The initial condition is $\mathrm{q}_0 = [x , y]^\intercal = [-3.0 , -3.0]^\intercal$. 
The goal is to drive the system to the origin. The workspace contains two obstacles of radius $0.2$ at coordinates $[-3.0 , -1.0]^\intercal$ and $[-2.0 , -2.0]^\intercal$. Matrix 
$\Sigma(\mathrm{q})$ is the $2\times2$ identity, and $R_i$ is chosen to 
satisfy $\|\gamma_{i-1} - {\gamma_i} \| < R_i -  2\varepsilon$ and $\min\{ \| \gamma_i - z \|, z \in \mathcal{O}\} > R_i$  with  $\varepsilon = 0.1$. A navigation function $Q(\mathrm{q})$ is 
constructed on $\mathbb{R}^2$ and  a trajectory for $\dot{\hat{\mathrm{q}}} = \hat{u}(\hat{\mathrm{q}}(t))$ is generated based on \cite{Tanner-10}. The simulation of the complete algorithm is shown in the Fig. \ref{fig:SimSMPC-unboundedinputs}. The navigation function is depicted in the form of a contour plot, while the discrete way-points are center of filled (red) circles. The boundaries $\mathcal{M}_i$ are chosen based on \eqref{condition1}-- \eqref{condition3} and are marked in the figure by dotted black circles.

\begin{figure}[!h]
\centering
\subfigure[Stochastic Path]{
          \includegraphics[width=0.35\textwidth]{path.pdf}
	\label{fig:SimSMPC}
}
\subfigure[Inputs]{
             \includegraphics[width=0.35\textwidth]{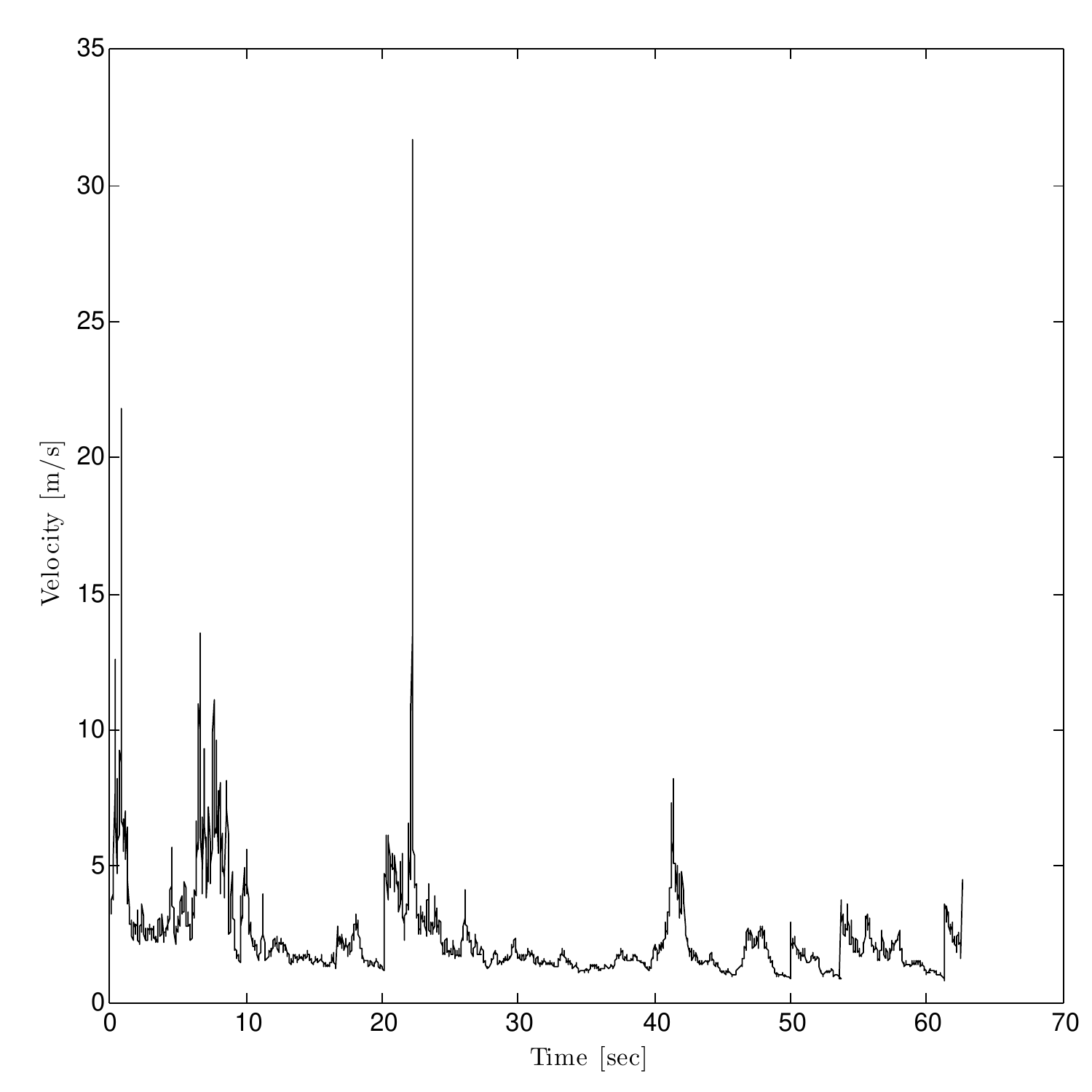}
	\label{fig:Input}
}
\caption[]{Simulation of a stochastic receding horizon control for a stochastic single integrator moving in a two obstacle environment. The simulation was generated using Euler-Maruyama method implemented in $\text{MATLAB}^{\copyright}$ Econometrics toolbox. \subref{fig:SimSMPC} The blue trajectory shows the actual stochastic path taken by the system. The initial condition of the system is marked with a black square. The black dashed circles represents the boundary $\mathcal{M}_i$ while red disks represent the region around way-points $\gamma_i$ with its boundary $\mathcal{N}_i$ and the blue circle is the boundary around the final goal. and \subref{fig:Input} Norm of the control inputs for the entire simulation with unbounded inputs.}
\label{fig:SimSMPC-unboundedinputs}
\end{figure}

\paragraph*{The effect of input saturation}
The following controller is a saturated version of \eqref{actualcontrol}:
\begin{equation}
\label{eq:saturatedcontrol}
\check{u}_i(\mathrm{q}) = - |u(\mathrm{q})|_{\mathsf{max}} \cdot \tanh \left(\frac{\mathrm{q} -\gamma_i }{{\big(R_i -\| \mathrm{q} -\gamma_i \|\big)}{\| \mathrm{q} -\gamma_i \|}} \right) \enspace.
\end{equation}
Figure \ref{fig:SimSMPC-boundedinputs} shows a sample path for the bounded input case, and
quantifies the norm of the inputs used.
\begin{figure}[thpb]
\centering
\subfigure[Stochastic Path]{
	\centering
	 \includegraphics[width=0.35\textwidth]{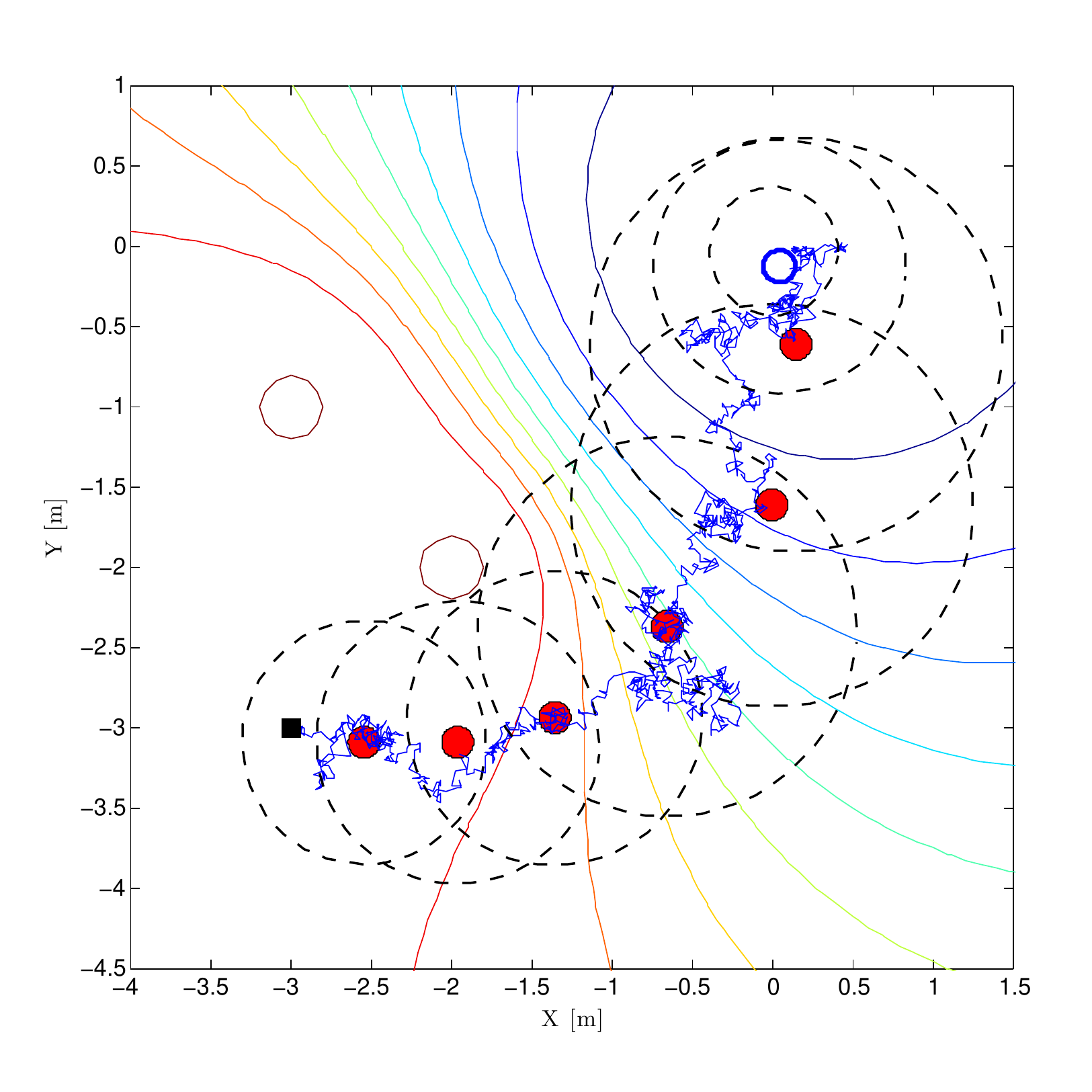}
	\label{fig:SimSMPC-bounded}
}
\subfigure[Inputs]{
	\centering
	 \includegraphics[width=0.35\textwidth]{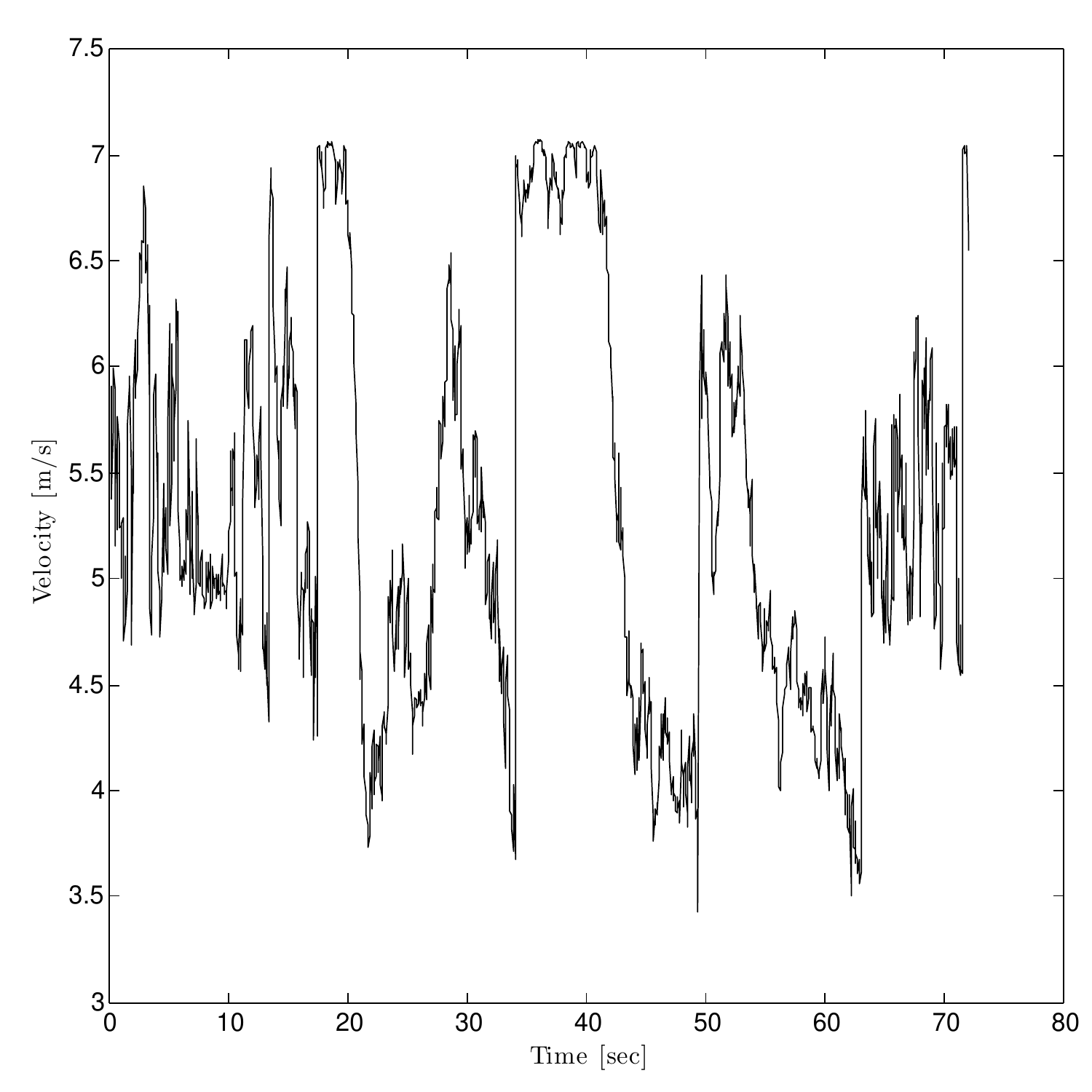}
	\label{fig:Saturated-Input}
}
\caption[Optional caption for list of figures]{Simulation of a stochastic receding horizon control for a stochastic single integrator moving in a two obstacle environment with bounded inputs. The system \eqref{eq:linearexample} was simulated with bounded inputs \eqref{eq:saturatedcontrol} and $|u(\mathrm{q})|_{\mathsf{max}} = 5$.  \subref{fig:SimSMPC-bounded} The blue trajectory shows the actual stochastic path taken by the system. The initial condition of the system was $[-3 , -3]^\intercal$ represented by a square. The black dashed circles represent the boundary $\mathcal{M}_i$ while red disks represent the region around way-points $\gamma_i$ with its boundary $\mathcal{N}_i$ and the blue circle is the boundary around the final goal. and \subref{fig:Saturated-Input} Norm of the saturated control inputs. Each component of the input was saturated at  $|u(\mathrm{q})|_{\mathsf{max}} = 5$ using \emph{tanh} function.}
\label{fig:SimSMPC-boundedinputs}
\end{figure}
 
 \begin{figure}[thpb]
      \centering
      \includegraphics[width=0.4\textwidth]{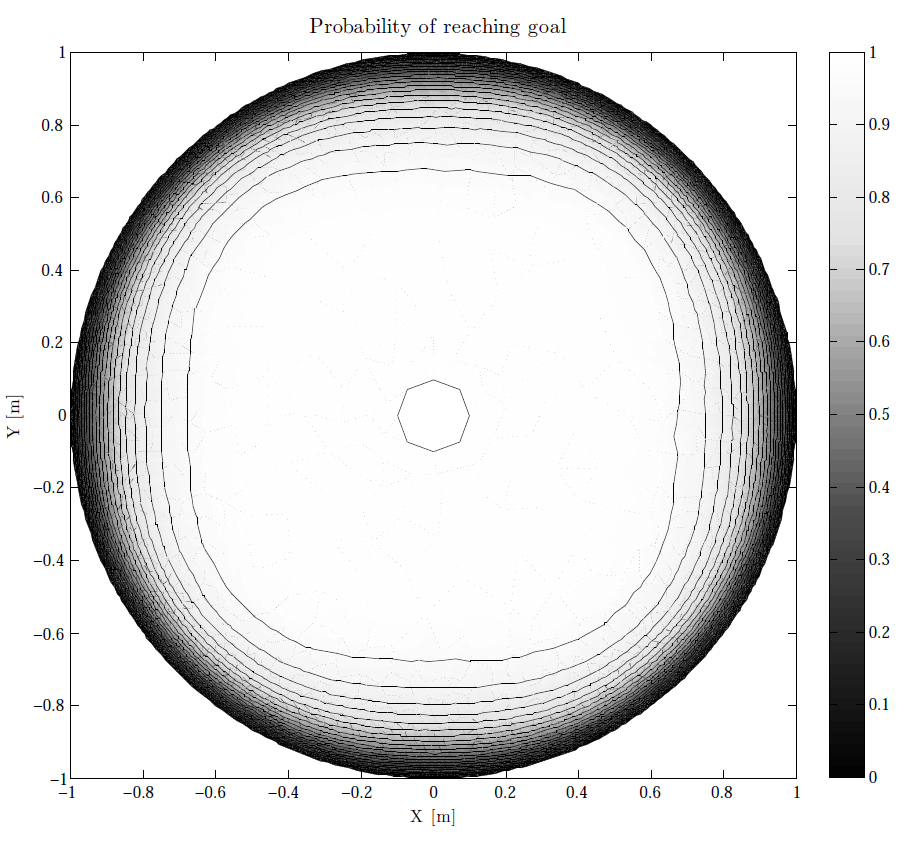}
      \caption{The probability of first hitting the goal boundary for the system \eqref{eq:examplesystem} using bounded input \eqref{eq:saturatedcontrol} with $|u(\mathrm{q})|_{\mathsf{max}} = 5$. The probability of reaching the desired boundary for each local controller can be computed according to \cite{shah-iros2011}.}
      \label{probtanh}
   \end{figure}  
As discussed earlier, bounded inputs \eqref{eq:saturatedcontrol} will not result in success with
 probability one (i.e.\ the probability of first hitting $\partial \mathcal{B}_{\gamma_i}({\varepsilon})$) and the probability of success for each local controller can be computed according to \cite{shah-iros2011}. Figure \ref{probtanh} represents the probability of hitting the goal boundary $\partial \mathcal{B}_{\gamma_i}({\varepsilon})$, before exiting the domain elsewhere for any given initial condition. It can be seen that there is always a nonzero probability that the system exits from $\partial \mathcal{B}_{\gamma_i}({R_i})$ instead of $\partial \mathcal{B}_{\gamma_i}({\varepsilon})$ under bounded inputs, and this probability becomes higher for initial conditions closer to $\partial \mathcal{B}_{\gamma_i}({R_i})$. 

To recover convergence under bounded inputs, we implement the recovery strategy. 
The implementation is shown in Fig. \ref{recovery-strategy}.    
We observe that the probability of convergence with recovery strategy can be one in absence of obstacles and sufficiently (infinitely) large outer boundary. In the presence of obstacles, the computation of the probability of convergence can only be approximated by a numerical estimation for finite way-points.\footnote{The probability of convergence can be shown to be equal to one if we consider the state constraints to be reflective boundary; this is a topic for a different paper.}
\begin{figure}[thpb]
      \centering
      \includegraphics[width=0.4\textwidth]{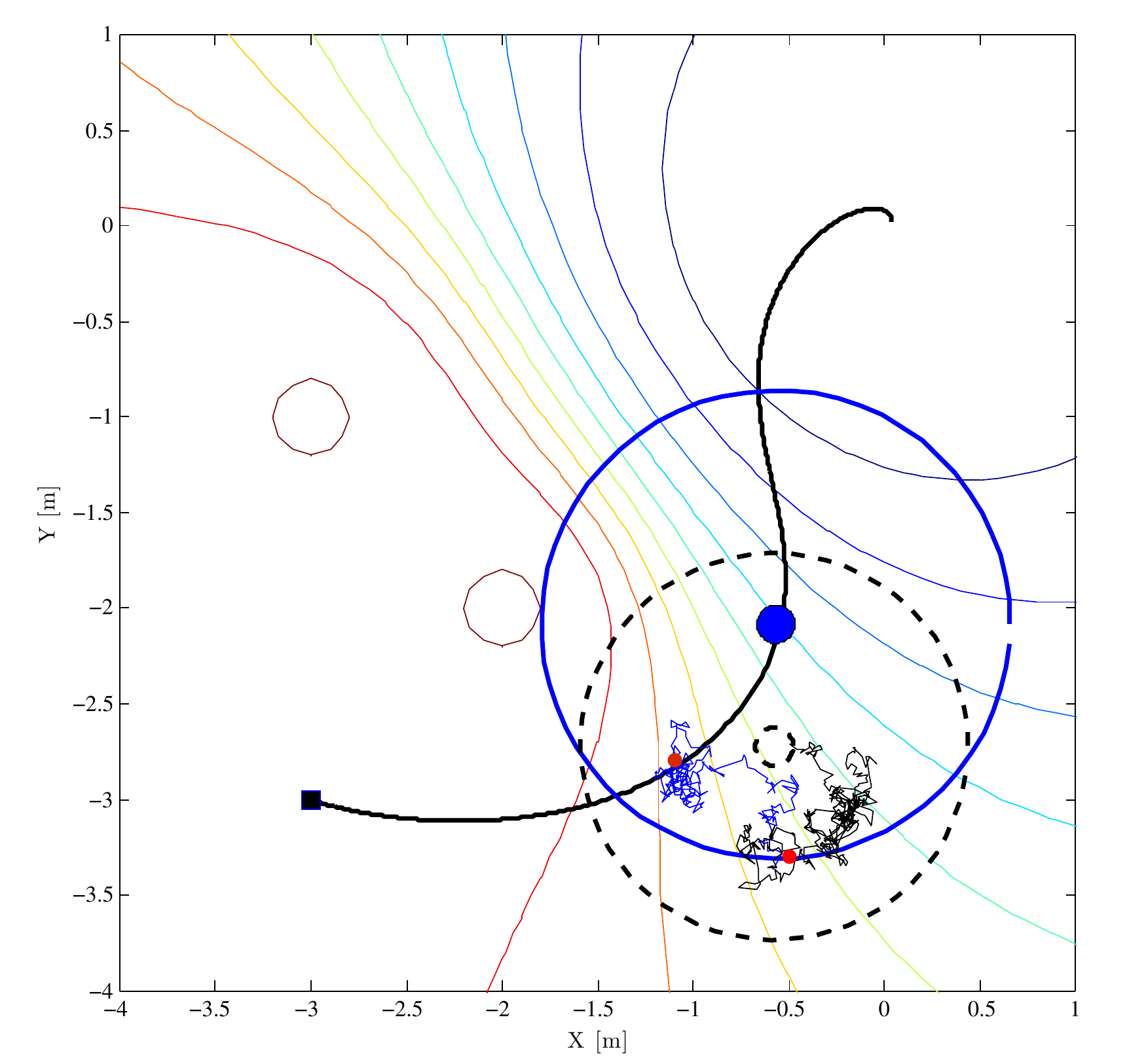}
      \caption{An example of the recovery strategy. The blue trajectory is evolution under a controller $\check{u}_i(\mathrm{q})$ which fails and the system exits at $\partial \mathcal{B}_{\gamma_i}({R_i})$. The dotted circles form domain of the recovery controller and the system is driven back inside the domain $\mathcal{D}_{i}$.}
      \label{recovery-strategy}
   \end{figure} 

\subsection{A Nonlinear System}

Finding a solution to the \ac{pde} \eqref{eq:Dirichlet} is central to the proposed control design. 
In Section~\ref{section:linearexample}, such a solution can be obtained explicitly, but with 
\eqref{eq:Dirichlet} having varying coefficients, this is not
true in general.  In this section we demonstrate a solution approach that is based on the 
Feynman-Kac formula.

\paragraph*{Problem formulation} 
Consider a mobile robot with three omni-directional wheels (Fig.~\ref{fig:omnirobot}).
In Fig.~\ref{fig:omnirobot}, $x$, $y$ mark the position, with respect an inertial $X$--$Y$ 
frame, of the local, body-fixed frame $X_m$--$Y_m$.   The orientation of the local frame
with respect to $X$--$Y$ is given by angle $\theta$.  The dynamical system modeling the
robot has as state the vector $\mathrm{q}  = [x, y, \theta]^\intercal$. 
The input to the system is a vector $u = [U_1, U_2, U_3]^\intercal$ of the linear velocities
of the three wheels, denoted $U_1$, $U_2$, $U_3$, respectively.  
Stochastic noise affects all three coordinates $x$, $y$ and $\theta$. 
The equations of motion for such a system can be represented by the following \ac{sde}
\begin{multline}
\label{eq:omnirobot-examplesystem}
\small \left[ \begin{array}{c}\dot x \\ \dot y \\ \dot \theta \end{array} \right] = 
  \left[ \begin{array}{ccc} \frac{2}{3}\cos(\theta + \delta) & -\frac{2}{3}\cos(\theta - \delta) & \frac{2}{3}\sin(\theta)\\ \frac{2}{3}\sin(\theta + \delta) & -\frac{2}{3}\sin(\theta - \delta) & -\frac{2}{3}\cos(\theta) \\ \frac{1}{3L} & \frac{1}{3L}& \frac{1}{3L} \end{array} \right]
 \left[ \begin{array}{c}{U}_1 \\ {U}_2 \\{U}_3 \end{array} \right] \\  + 
\small \left[ \begin{array}{ccc}0.2 & 0 & 0 \\ 0 & 0.2 & 0 \\ 0 & 0 & 0.2 \end{array} \right]
 \left[ \begin{array}{c}dW_1 \\ dW_2 \\dW_3 \end{array} \right]
\end{multline}
\begin{figure}[thpb]
      \centering
      \includegraphics[width=0.4\textwidth]{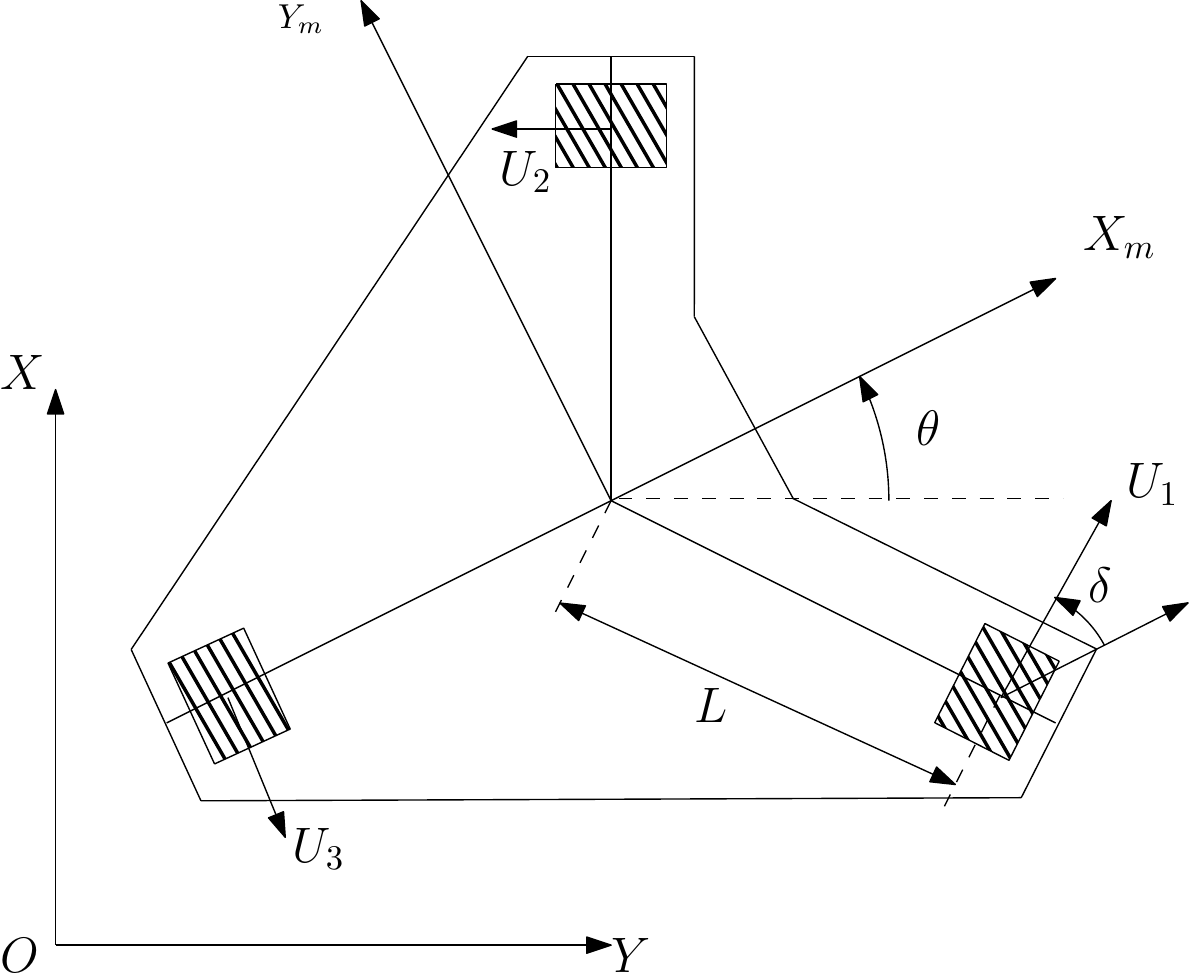}
      \caption{A graphical representation of an omni-directional robot, showing the variables involved in
      the dynamical model \eqref{eq:omnirobot-examplesystem}.}
      \label{fig:omnirobot}
   \end{figure} 
\begin{remark}
\label{remark:r4}
Formally, $\mathrm{q}  = [x, y, \theta]^\intercal$ belongs in the two-dimensional special Euclidean
group $\mathrm{SE}(2)$; it can, however, be embedded in $\mathbb{R}^4$  \cite{Lav06}, 
where the usual metrics can be used.   Here, the metric 
$\|[x_1, y_1, \theta_1]^\intercal\| = \sqrt{x_1^2 + y_1^2 + (\cos \theta_1- 1)^2 + (\sin \theta_1)^2}$
(see \cite{Lav06}) is used.
%
%
\end{remark}

The goal is to a find control law $U_i\big(\mathrm{q}(t)\big)$ to drive 
\eqref{eq:omnirobot-examplesystem} to the origin $x=y=\theta=0$,
using inputs of minimal magnitude, following paths of minimal length, and avoiding obstacles 
along the way.  The robot's workspace is a torus, containing a finite number $M$ of 
torus-shaped obstacles at locations $\mathrm{q}_j$, $j=1,2,\dots,M$.
The robot's outer workspace boundary, and those of the
obstacles for $i=1,\dots,M$ are is defined as 
\begin{subequations}
\label{modified boundaries}
\begin{align}
\partial \mathcal{S} &\triangleq \{(x,y,\theta) \in \mathbb{R}^2\times\mathbb{S} \mid x^2 + y^2 =  \rho_0^2, \forall \theta \in \mathbb{S}\} \\
\partial \mathcal{O}_i &\triangleq \{(x,y,\theta) \in \mathbb{R}^2\times\mathbb{S} \mid (x-x_i)^2 + (y - y_i)^2 =  \rho_i^2, \forall \theta \in \mathbb{S}\} \ .
\end{align}
\end{subequations}

Matching \eqref{eq:omnirobot-examplesystem} to \eqref{localSDE1} we identify the different terms
as follows:
\begin{align*}
b(\mathrm{q}) & = [0 \  0 \  0]^T, &  \\
G(\mathrm{q}) & =    \begin{bmatrix} \frac{2}{3}\cos(\theta + \delta) & -\frac{2}{3}\cos(\theta - \delta) & \frac{2}{3}\sin(\theta)\\ \frac{2}{3}\sin(\theta + \delta) &  -\frac{2}{3}\sin(\theta - \delta) & -\frac{2}{3}\cos(\theta) \\ \frac{1}{3L} & \frac{1}{3L}& \frac{1}{3L} \end{bmatrix}, 
& \\
\Sigma(\mathrm{q}) &= \begin{bmatrix} 
0.2 & 0 & 0 \\ 0 & 0.2 & 0 \\ 0 & 0 & 0.2 \end{bmatrix}. 
\end{align*}

\paragraph*{Deterministic Path Planning} 
Using the metric introduces in Remark~\ref{remark:r4}, and the definition of obstacle and 
outer boundary in \eqref{modified boundaries}, we apply the path planning approach of 
Section \ref{section:linearexample}, selecting a fixed $R$ satisfying 
$\inf_{z \in \mathcal{O}, t>0} \| \hat{\mathrm{q}}(t) - z \| > R > 2\varepsilon> 0$.

Let us denote $\hat{\mathrm{q}}^*(t)$ the obstacle-free continuous state trajectory found 
using, say \cite{Tanner-10}.  Then the path is expressed directly as
$
\Gamma_T \triangleq \{ \gamma \in \mathbb{R}^2\times\mathbb{S} \mid \exists t \in \mathbb{R} ; \gamma = \hat{\mathrm{q}}^*(t)\} 
$.

\paragraph*{Way-point Generation} 
Here we will select a sequence $\{\gamma_i\}_{i=0}^{N} \in \Gamma_T$, of
waypoints. The objective of stochastic controller for each discrete state $i$ is to 
make \eqref{eq:omnirobot-examplesystem} converge $\varepsilon>0$ close to 
way-point $\gamma_i$.

To this end, define a set $\overline{\mathcal{B}}_{\gamma_i}({\varepsilon}) \triangleq \{\mathrm{q} \in 
\mathcal{P} : \| \mathrm{q} - \gamma_i\| \leq \varepsilon \}$ and denote its boundary 
$\partial \mathcal{B}_{\gamma_i}({\varepsilon})$.
Then define domains $\mathcal{D}_i = \{(x,y,\theta) \in \mathbb{R}^2\times\mathbb{S} \mid x^2 + y^2 < R,
\enVert{(x,y,\theta)} > \varepsilon, \forall \theta \in \mathbb{S}\}$, and select an arbitrary set of $N$
points from $\Gamma_T$, such that 
$\gamma_0 = \hat{\mathrm{q}}(t_0)$, $\gamma_N = \hat{\mathrm{q}}^*(T)$, and for $i=1,\ldots,N-1$,
\begin{align}
\label{condition1}
&\max_{a \in \overline{\mathcal{B}}_{\varepsilon_i}}\{Q(a)\} - \min_{b \in \overline{\mathcal{B}}_{\varepsilon_{i-1}}}\{Q(b)\} \leq -\eta(\|\gamma_{i-1}\|) \\
\label{condition2}
&R-2\varepsilon > \|\gamma_{i-1} - \gamma_i \| > 2\varepsilon
\end{align}
The boundaries $\mathcal{N}_i$
and  $\mathcal{M}_i$ are defined as $\mathcal{N}_i = \partial \mathcal{D}_i \cap \partial \mathcal{B}_{\gamma_i}({\varepsilon}) $ and $\mathcal{M}_i = \partial \mathcal{D}_i \setminus \mathcal{N}_i$, respectively
for all $i=1,\ldots,N$. 

\paragraph*{Stochastic optimal controller} 
The \ac{pde} \eqref{eq:Dirichlet} is now written as
\begin{align}\label{dirichlet-omnirobot}
\mathcal{L} g &= 0 &&\text{in } \mathcal{D}_i\\
\nonumber g &= \exp\!\big(-\Phi(\xi(\tau_{\mathcal{N}_i}))\big) & &
\text{on } \mathcal{M}_i \cup \mathcal{N}_i =  \partial D_i
\end{align}
where $\mathcal{L}$ is an operator on functions defined as 
$
\mathcal{L}(\cdot) = 
\frac{1}{2} \mathrm{tr}\left\{ \partial_{\mathrm{q}\mathrm{q}}(\:\cdot\:) \,G(\mathrm{q})\,
\Sigma(\mathrm{q})\,\Sigma^\intercal(\mathrm{q})\,G^\intercal(\mathrm{q})\right\}
$.

Equation \eqref{dirichlet-omnirobot} does not admit analytic solutions.   
Common applicable numerical methods such as finite differences and finite elements have
difficulty producing acceptable solutions for instances of problems with dimension larger
than three and complex boundary conditions. 
Alternatively, the Feynman-Kac's formula (see Section \ref{section:dayscontroller}), relates 
the \ac{pde} to an \ac{sde}:
\begin{equation}
\tiny{
\label{eq:feynman-system}
\begin{bmatrix} 
\dif \xi_1 \\ \dif \xi_2 \\ \dif \xi_3 
\end{bmatrix}  = 
 \begin{bmatrix} 
\textstyle{\frac{2}{3}}\cos({\xi_3} + \delta) & -\textstyle{\frac{2}{3}}\cos({\xi_3} - \delta) & 
\textstyle{\frac{2}{3}}\sin{\xi_3}\\ \textstyle{\frac{2}{3}}\sin({\xi_3} + \delta) & 
-\textstyle{\frac{2}{3}}\sin({\xi_3} - \delta) & -\textstyle{\frac{2}{3}}\cos{\xi_3} \\ 
\textstyle{\frac{1}{3L}} & \textstyle{\frac{1}{3L}} & \textstyle{\frac{1}{3L}} 
\end{bmatrix} 
\:
\begin{bmatrix}
0.2 & 0 & 0 \\ 0 & 0.2 & 0 \\ 0 & 0 & 0.2 
\end{bmatrix} 
\begin{bmatrix} 
\dif W_1 \\ \dif W_2 \\ \dif W_3 
\end{bmatrix} }
\end{equation}
which is essentially the unforced system \eqref{eq:omnirobot-examplesystem}.
Then, we know that the function $g(\mathrm{q})$ satisfies
\begin{equation}
\label{probfunction}
g(\mathrm{q}) = \mathbb{P} \big[ \xi_t(t_i) \in \mathcal{N}_{i} | \xi_t = \mathrm{q} \big]
\end{equation}
where $t_i$ is the first exit time from the domain $\mathcal{D}_i$.

\paragraph*{Problem instantiation and simulation results}
The probability in \eqref{probfunction} can be estimated numerically\footnote{The source code to compute function $g(\mathrm{q})$ is available at http://code.google.com/p/stochastic-receding-horizon-control/} by 
simulating sufficiently many sample paths 
of \eqref{eq:feynman-system} with different initial conditions $\mathrm{q}$.  
We produce these sample paths using the Euler-Maruyama method \cite{SDEsimulation}. 
Using the same method, we also obtain sample paths for \eqref{eq:omnirobot-examplesystem}. 
A $41\times41\times41$ grid is imposed on the state space, and treating each node as an initial 
condition, we produce  500 sample paths and estimate \eqref{probfunction}.
With the estimate of \eqref{probfunction}, the control law is computed numerically as
\begin{equation*}
u^*_i(\mathrm{q}) = - \Sigma(\mathrm{q}) \, \Sigma^\intercal(\mathrm{q})\, G^\intercal(\mathrm{q}) \, \partial_{\mathrm{q}}\!\big(- \log(g(\mathrm{q}))\big)\  .
\end{equation*}
Figure~\ref{fig:gx} presents two numerical approximations of $g(\mathrm{q})$ in the form of 2D colormaps 
with robot orientation set at $0$ and $\frac{\pi}{2}$ radians, respectively.  Equipped with such a 
map, a numerical gradient can be used to calculate the control input. 
Figure \ref{fig:omni_path} shows a single sample path for the closed loop version of 
\eqref{eq:omnirobot-examplesystem}.  The time history of individual states $x$, $y$ and $\theta$
are shown in Figs.~\ref{fig:omni_x_traj}--\ref{fig:omni_theta_traj},
indicating the convergence to an $\varepsilon$ neighborhood of the origin. 
Figure~\ref{fig:omni_inputs} plots the norm of the control inputs used.  Numerical data
confirmed that the probability that the closed loop system hits every desired goal boundary 
$\partial \mathcal{N}_i$ is one.
\begin{figure}[thpb]
\centering
\subfigure[Solution $g(\mathrm{q})$ for robot orientation $\theta = 0$]{
	\centering
	 \includegraphics[width=0.35\textwidth]{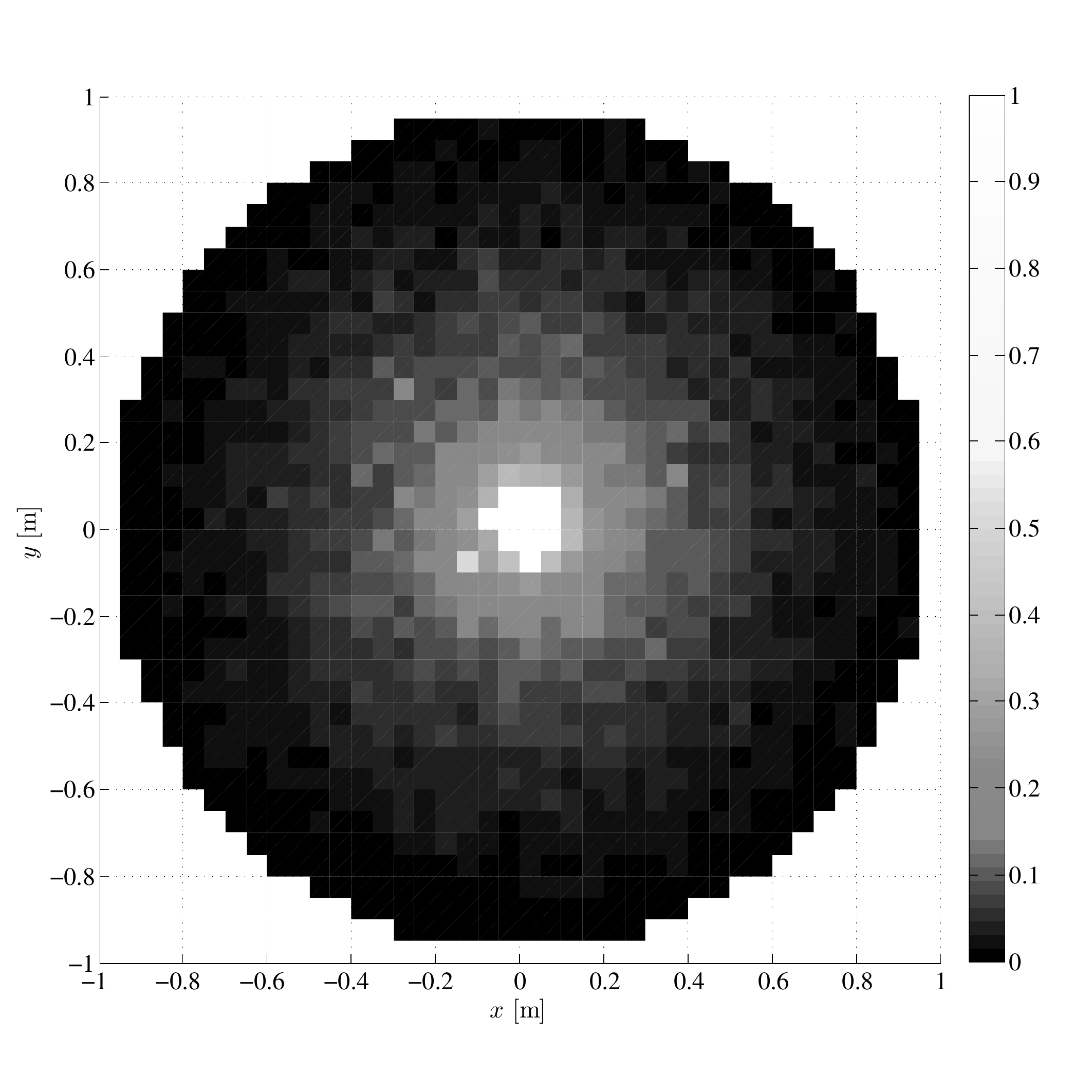}
	\label{fig:gx_theta0}
}
\subfigure[Solution $g(\mathrm{q})$ for robot orientation $\theta = \frac{\pi}{2}$]{
	\centering
	 \includegraphics[width=0.35\textwidth]{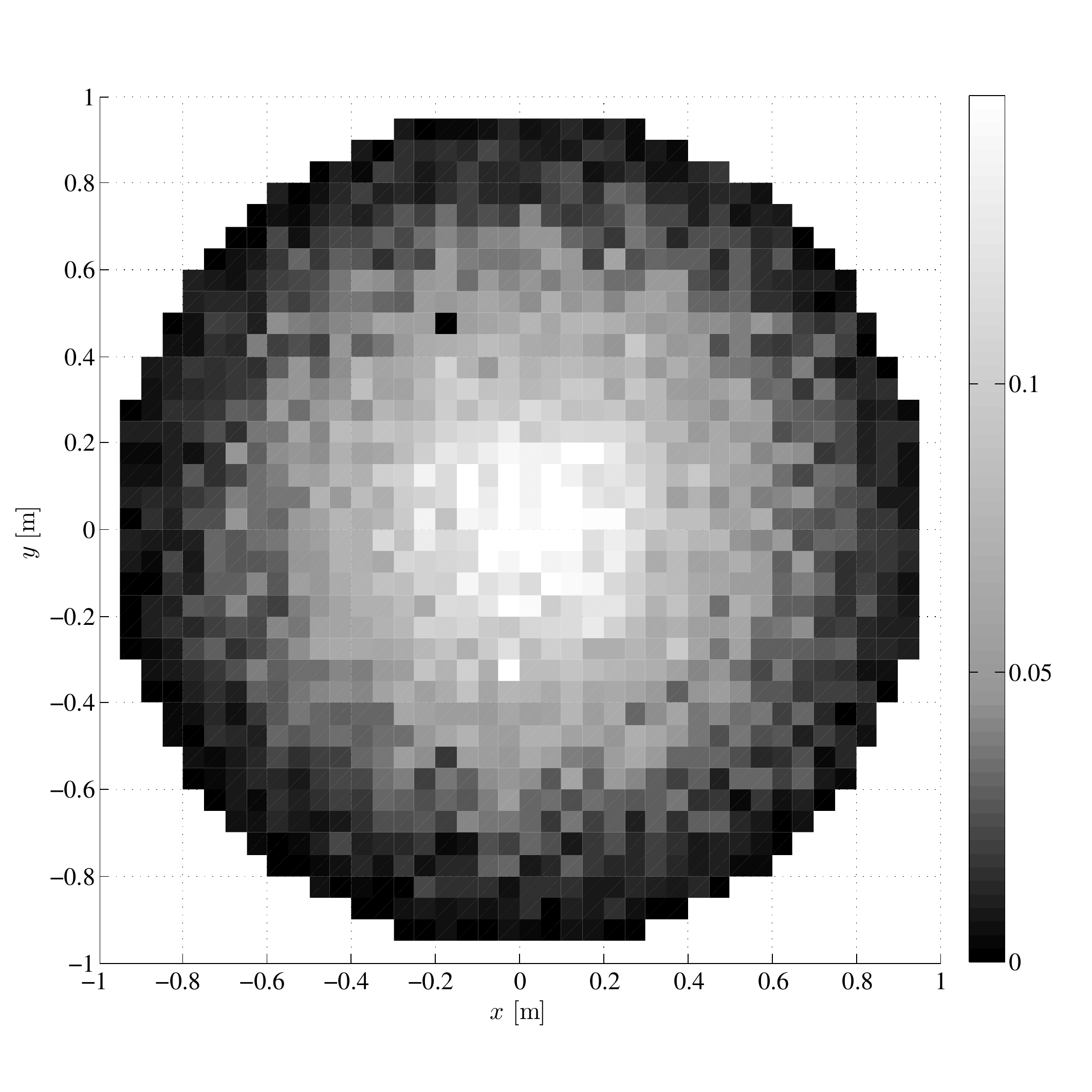}
	\label{fig:gx_theta90}
}
\caption[]{Numerical solution $g(\mathrm{q})$ of \ac{pde} \eqref{dirichlet-omnirobot}
for stochastic system \eqref{eq:omnirobot-examplesystem} for $R=1$ and $\varepsilon = 0.1$.}
\label{fig:gx}
\end{figure}

\begin{figure}[thpb]
      \centering
      \includegraphics[width=0.4\textwidth]{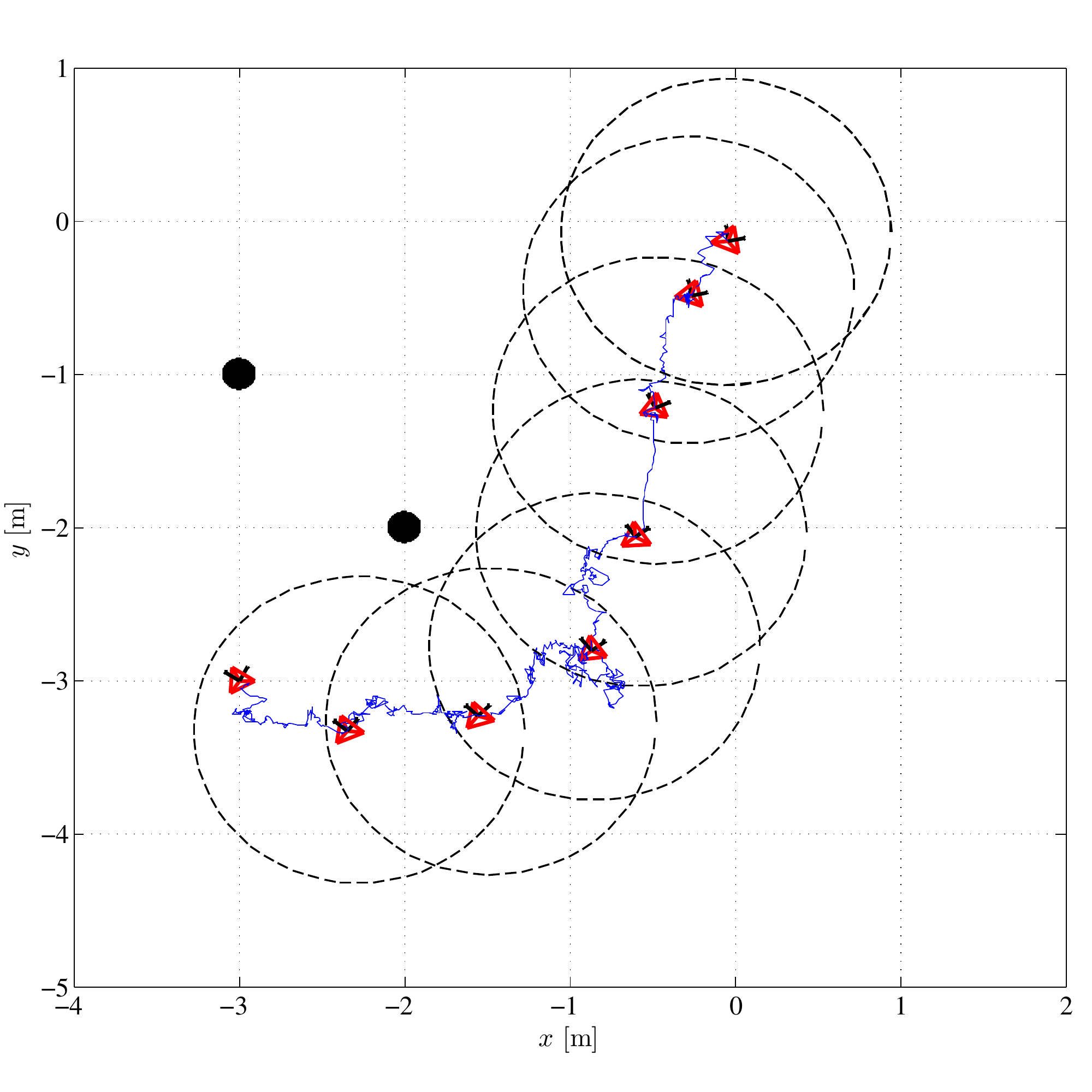}
      \caption{A sample path for initial condition $[x,y,\theta]^\intercal = [-3.0,-3.0,1.0]^\intercal$. Black circular dots represent two obstacles at $[-3,-1,\star]^\intercal$ and $[-2,-2,\star]^\intercal$. The robot position is shown by a red triangle and local coordinate axis at each switching point. Dotted circles represent the projection of boundary $\mathcal{M}_i$ on the $X$-$Y$ plane.}
      \label{fig:omni_path}
   \end{figure} 

\begin{figure}[thpb]
\centering
\subfigure[$x$ trajectory]{
	\centering
	 \includegraphics[width=0.35\textwidth]{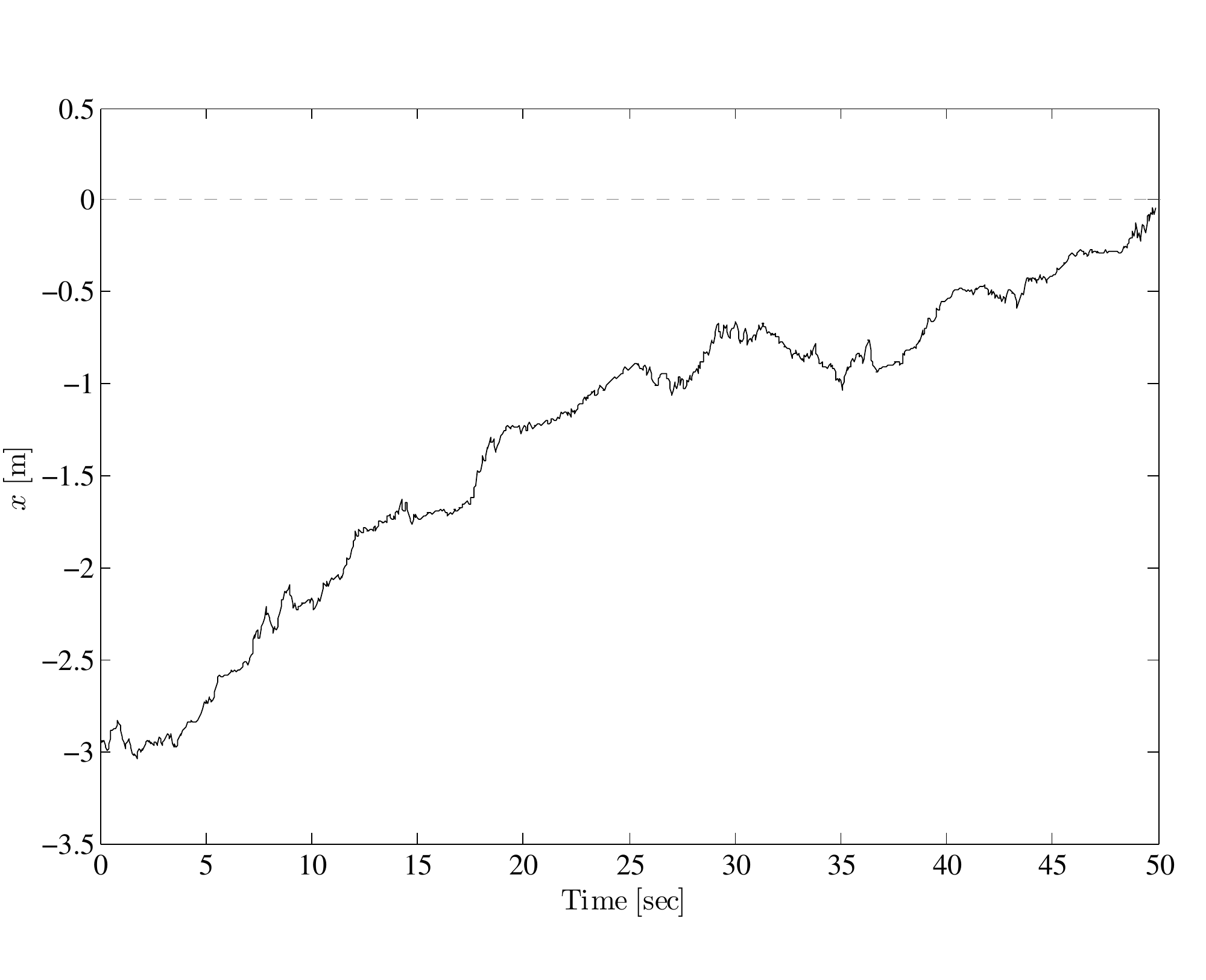}
	\label{fig:omni_x_traj}
}
\subfigure[$y$ trajectory]{
	\centering
	 \includegraphics[width=0.35\textwidth]{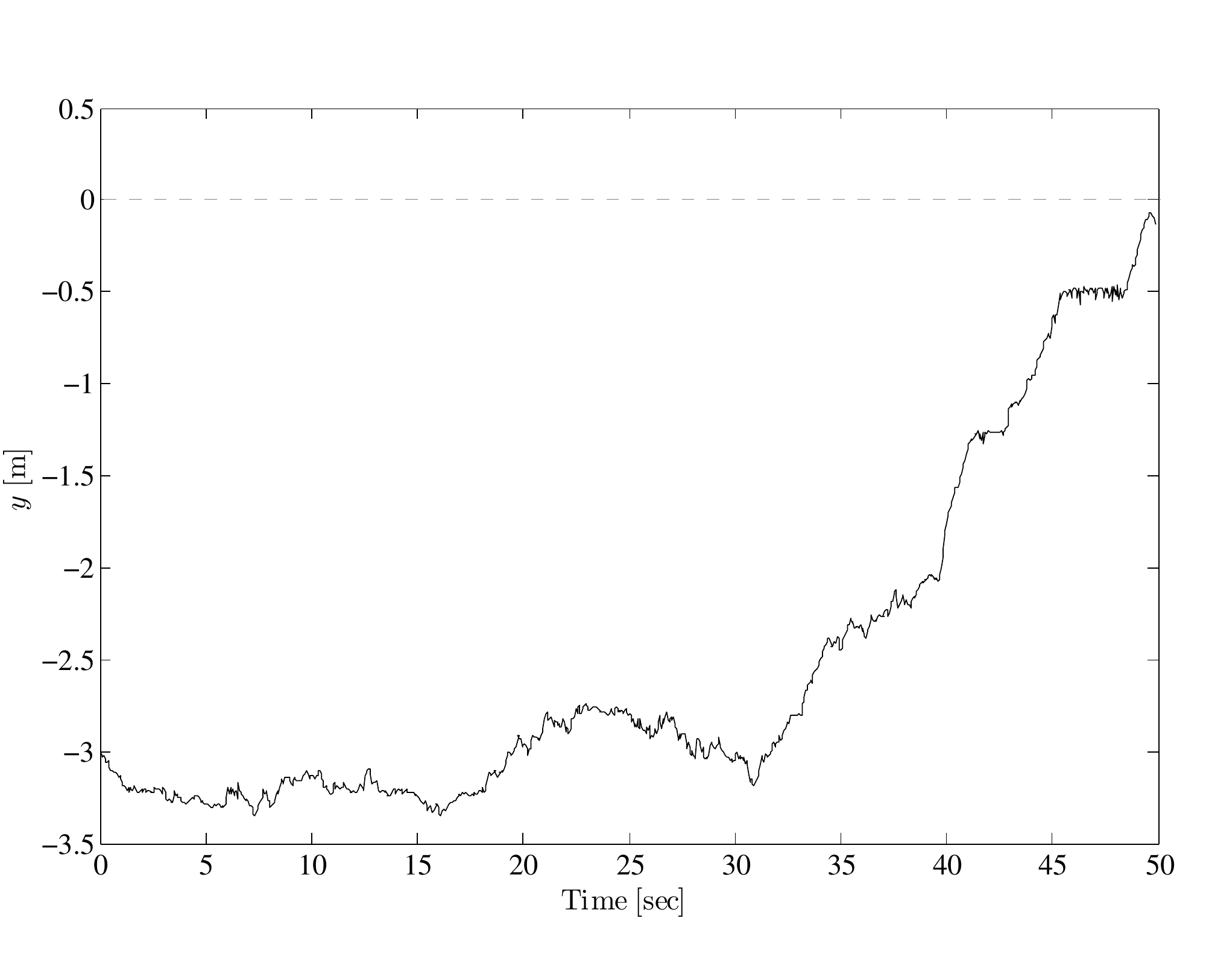}
	\label{fig:omni_y_traj}
}
\subfigure[$\theta$ trajectory]{
	\centering
	 \includegraphics[width=0.35\textwidth]{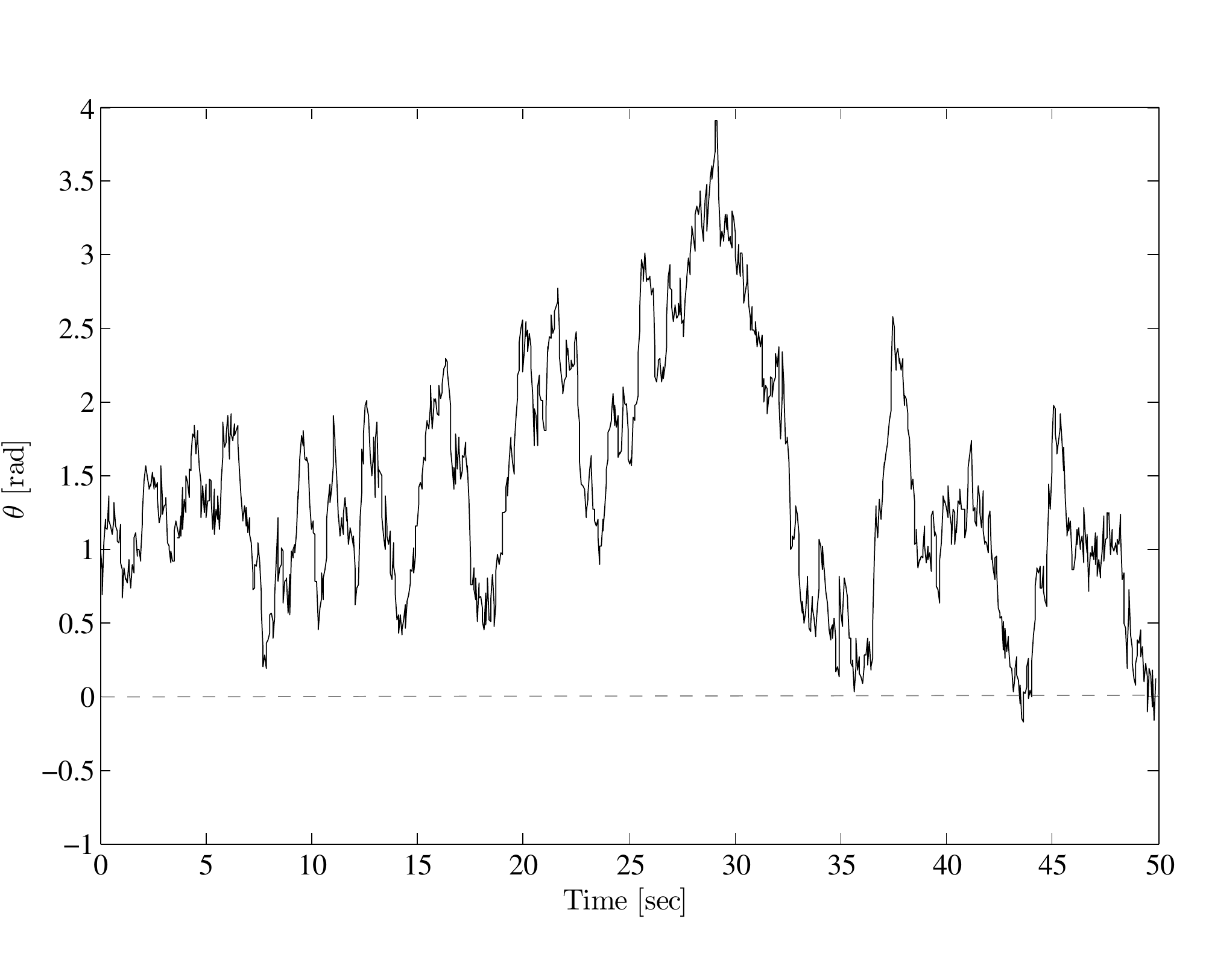}
	\label{fig:omni_theta_traj}
}
\subfigure[Norm of inputs]{
	\centering
	 \includegraphics[width=0.3\textwidth]{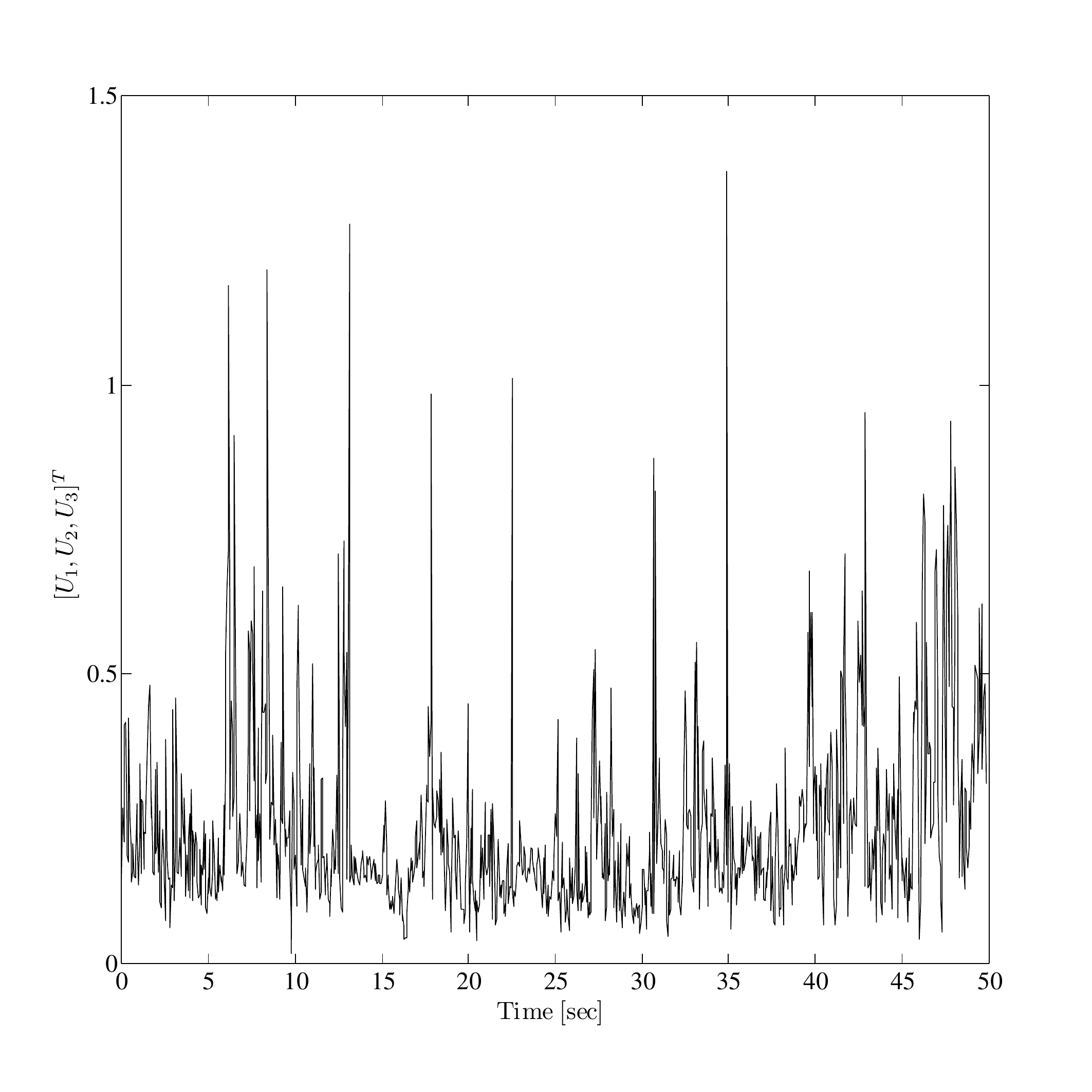}
	\label{fig:omni_inputs}
}
\caption[]{Individual trajectories and control input for the simulation presented in Fig. \ref{fig:omni_path}. \subref{fig:omni_x_traj}, \subref{fig:omni_y_traj}, \subref{fig:omni_theta_traj} Individual state trajectories converges to zero. \subref{fig:omni_inputs} The Euclidean norm of the control input applied for the entire trajectory.}
\label{fig:Unicycle-Results}
\end{figure}
%
\section{Conclusions}
\label{section:conclusions}
The proposed method allows the design of a receding horizon navigation controller for nonlinear systems governed by stochastic differential equations. 
If a feasible path, optimal or otherwise, is available in the form of a finite sequence of way-points, then 
an an optimal control law can be found to steer the stochastic system between these way-points, while keeping it close to the path and away from unsafe regions with probability one.  In cases where
control inputs are forced within upper and lower bounds, and state constraints (obstacles) are 
imposed, almost-sure convergence and safety is impossible, but it can be achieved with some
probability which depends on how severe the input bounds are compared with respect to the
magnitude of subjected noise.  For nonlinear systems with dynamics not permitting 
analytic solutions for the resulting \ac{pde}s, numerical solutions for dimensions up to $5$ or $6$ 
are shown to be well within the reach of currently available computing platforms.

\bibliographystyle{IEEEtran}
\bibliography{StocOptControl}

\end{document}